\documentclass[a4paper,12pt]{amsart}
\usepackage[utf8]{inputenc}
\usepackage[margin=1in]{geometry}
\usepackage{graphicx}
\usepackage{amssymb,amsmath,amsthm}
\usepackage{tikz}
\usepackage{xcolor}
\usepackage[colorlinks=true,linkcolor=red,citecolor=blue]{hyperref}
\usepackage{cleveref}
\usepackage[foot]{amsaddr}
\usepackage[style=alphabetic]{biblatex} 
\newtheorem{thm}{Theorem}[section]
\newtheorem{prop}[thm]{Proposition}
\newtheorem{lem}[thm]{Lemma}

\theoremstyle{definition}

\newtheorem{example}[thm]{Example}
\newtheorem{rem}[thm]{Remark}

\newcommand{\orcidicon}[1]{\href{https://orcid.org/#1}{\includegraphics[height=2.5ex]{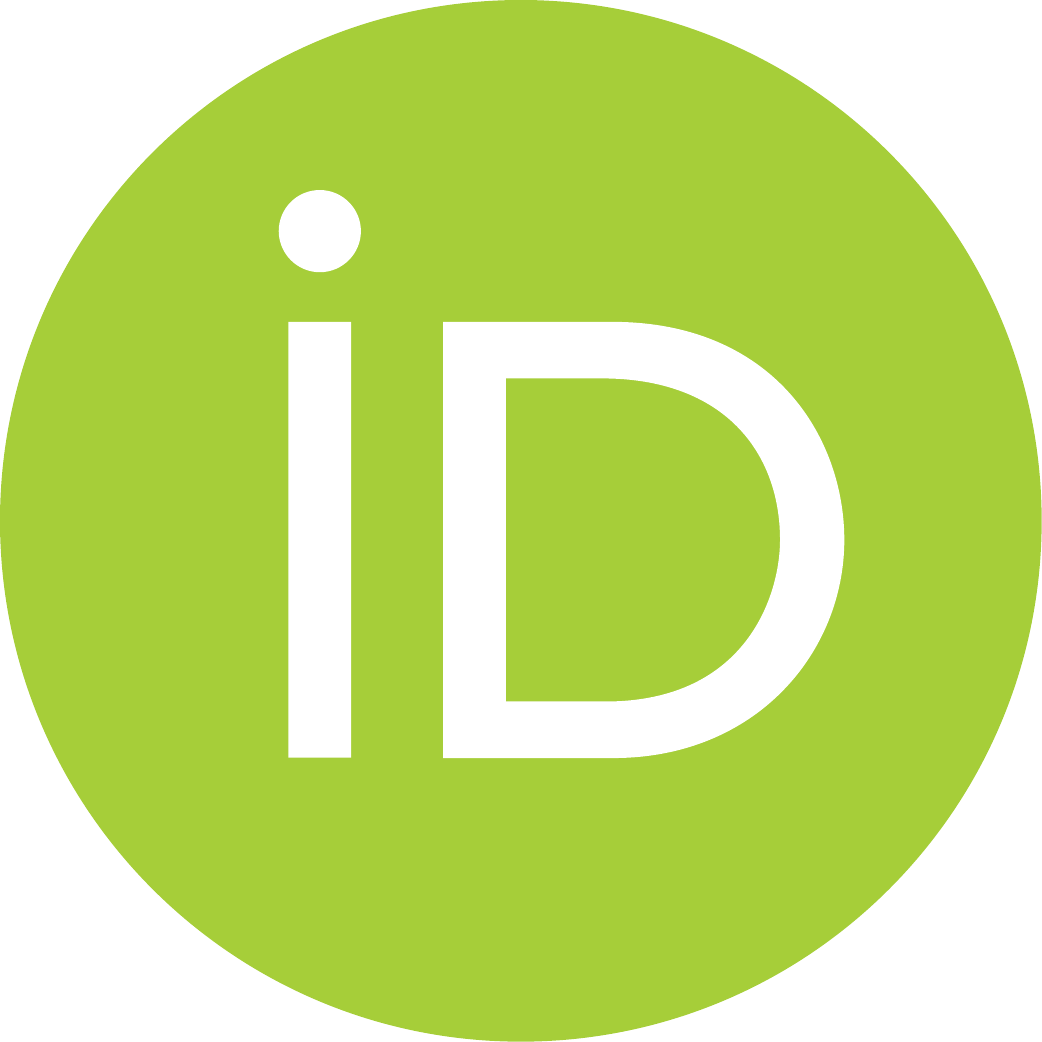}}}
\newcommand{\iu}{\mathrm{i}}
\newcommand{\inv}[1]{ {#1}^{-1} }
\newcommand{\m}[1]{\mathcal{#1}}
\newcommand{\mb}[1]{\mathbb {#1}}
\newcommand{\be}{\begin{equation}} 
\newcommand{\ee}{\end{equation}} 
\newcommand{\ti}[1]{\widetilde {#1}}

\newcommand{\ol}[1]{\overline{#1}}

\newcommand{\set}[1]{\{#1\}}
\newcommand{\din}{\ \dot{\in} \ }

\newcommand{\mg}{\m {C}_{\text{max}}(T)}

\title{Schwinger, ltd: Loop-tree duality in the parametric representation}
\author{Marko Berghoff \protect\orcidicon{0000-0002-9108-3045}}
\email{berghoffmj@gmail.com}
\address{Mathematical Institute, University of Oxford, Oxford, UK (now at: Institut für Mathematik, Humboldt-Universit\"at zu Berlin, Berlin, Germany)}

\begin{document}

\begin{abstract}
We derive a variant of the loop-tree duality for Feynman integrals in the Schwinger parametric representation. This is achieved by decomposing the integration domain into a disjoint union of cells, one for each spanning tree of the graph under consideration. Each of these cells is the total space of a fiber bundle with contractible fibers over a cube. Loop-tree duality emerges then as the result of first decomposing the integration domain, then integrating along the fibers of each fiber bundle. 

\noindent
As a byproduct we obtain a new proof that the moduli space of graphs is homotopy equivalent to its spine. In addition, we outline a potential application to Kontsevich's graph (co-)homology.

\end{abstract}

\dedicatory{F\"ur Niels Andersen, anstelle der Sonnenuhr.}

\maketitle

\section{Introduction}

\subsection{Background}
Loop-tree duality expresses the momentum space integral of a Feynman diagram $G$ as a linear combination of (simpler) integrals indexed by the set of spanning trees of the underlying graph. This simplification is achieved by iterated applications of Cauchy's residue theorem. We sketch the idea, in order to motivate the alternative approach presented down below. For the general story see the review article \cite{ltd-review} as well as the references therein. We follow the arguments of \cite{causality-and-ltd}. 

Consider the Feynman integral of a one-loop graph on $n$ edges in a scalar, massive theory,
\be\label{eq:1loopintegral}
I_G=\int d^{D}k \prod_{i=1}^n \frac{1}{L_i^2(k,p)-m_i^2 + i\varepsilon},
\ee
where each $m_i>0$ and each $L_i$ is $\mb Z$-linear in the loop momentum $k\in \mb R^{1,D-1}$ and the external momenta $p_1,\ldots,p_n\in \mb R^{1,D-1}$. The precise form of the $L_i$ depends on a choice of momentum flow, but the integral is independent of this by momentum conservation.

Let us abbreviate the integrand by
\[
f(k,p)= \prod_{i=1}^n \frac{1}{L_i^2(k,p)-m_i^2 + i\varepsilon}.
\]
 We can do the integration over the energy component $k_0$ in \cref{eq:1loopintegral} by using the residue theorem, 
\[
  \int d k^0  f(k,p) =\oint_\Gamma d k^0 f(k,p)  = 2\pi \iu \mathrm{Res}(f,\Gamma),
\]
where $\Gamma$ is a curve that runs first along the real axis from $-\infty$ to $\infty$, then closes along a semi-circle in the upper/lower half-plane through the point $\pm i \infty$ (either choose one or average over both choices). For generic values of $p$ this will collect $n$ different residues, one for each pole $\rho_i=\rho_i(\vec{k},p): k^0=(L_i(k,p))^0\pm \sqrt{ L_i^2(k,p) +m_i^2 -i \varepsilon }$ of the propagators $L_i^2(k,p)-m_i^2 + i\varepsilon$, $i=1,\ldots,n$. See also the derivation in \cite[\S 7]{mb-dk} using divided differences to systematically extract the poles of each propagator.

This transforms $I_G$ into a linear combination of integrals
\be \label{eq:LTDmomentum}
I_G=\int d^Dk \prod_{i=1}^n \frac{1}{L_i^2(k,p)-m_i^2 + \iu\varepsilon} =  \sum_{i=1}^n \pm2\pi \iu\int d\vec{k} \ \mathrm{Res}(f,\Gamma,\rho_i)
\ee
indexed by the spanning trees of $G$ (each pole belongs to a propagator/edge, its complement is a spanning tree of $G$).

It seems natural to generalize this to the multi-loop case,
\be\label{eq:moreloopintegral}
I_G=\int \prod_{i=1}^\ell d^D k_i \prod_{j=1}^n \frac{1}{L_j^2(k,p)-m_j^2 + i\varepsilon},
\ee
using induction on the number of loops $\ell=h_1(G)$. However, the location of each pole depends not only on $k$ and $p$, but also on $\varepsilon$. It is therefore not immediately clear how to compute iterated residues of this form. We refer to \cite{ltd-review} for a thorough discussion and general solution of this problem.

 \subsection{Parametric space and sector decompositions}
 One may wonder whether a version of this loop-tree duality exists for Feynman integrals in Schwinger variables. 

 In this case there seems to be no straightforward analytical method to relate a parametric Feynman integral to a linear combination of integrals indexed by the spanning trees of $G$. 
Of course one could use the Schwinger trick to transform an integral of the form \eqref{eq:LTDmomentum} into a parametric version. However, if one wishes to work exclusively in parametric space, then a more combinatorial approach can be based on Hepp's \textit{sector decomposition} \cite{Hepp66,Speer75}. For a detailed account we refer to \cite{HEINRICH_2008}; see also \cite[\S 4.3]{kreimeryeats-interplay} as well as \cite{Schultka:toric} which discusses general sector decompositions from the viewpoint of toric geometry. This approach leads to formulae of the sought-after form, but there seems to be no general account of loop-tree duality for parametric Feynman integrals in the literature, at least not under this name.

\subsection{An alternative approach}
We propose here a different, more geometric solution. The idea was sketched in \cite{mb-dk}. It is based on a study of the combinatorial geometry of the integration domain. Recall that a parametric Feynman integral arises from integrating a differential form $\omega_G$ (defined in \cref{ssec:parametricFI}) over a simplex in projective space,
\be \label{eq:FIparametric}
I_G=\int_{\Delta_G  } \omega_G \ \text{ where } \Delta_G=\set{ [x_1:\ldots :x_n] \mid x_i > 0  }.
\ee

We will see below that (up to a subset of measure zero) the integration domain $\Delta_G$ decomposes into a disjoint union of cells $D_T$, indexed by the set of spanning trees of $G$, such that each cell is the total space of a fiber bundle $\pi_T$ over a cube $C_T=(0,1)^{|E_T|}$. 

The construction of the map $\pi_T$ can be sketched as follows. The set of all subgraphs of $G$ is partially ordered by inclusion. The \emph{order complex} of this poset\footnote{The order complex of $P$ is the simplicial complex whose vertices are the elements of $P$ and $p_0,\ldots,p_n \in P$ form an $n$-simplex iff $p_0\leq \ldots \leq p_n$.} is isomorphic to the first barycentric subdivision of $\ol \Delta_G = \set{ [x_1:\ldots :x_n] \mid x_i \geq 0  }$. We call the subcomplex formed by the vertices $w_\gamma$, where $\gamma=G/F$ for $F\subset G$ a forest, the \emph{spine} of $\Delta_G$. The other vertices represent the faces of $\ol \Delta_G$ where the set $\set{e\in E_G \mid x_e=0}$ determines a subgraph $\gamma$ with $h_\gamma=\dim H_1(\gamma,\mb Q)>0$. In the context of moduli spaces of graphs these faces are called \emph{faces at infinity} (see \cref{ssec:modspace}). Each spanning tree $T\subset G$ represents a maximal cell $C_T$ in the spine of $\Delta_G$. It is a $t=|E_T|$ dimensional cube, the union of all simplices on vertex sets $\set{w_G,w_{G/T_1},\ldots,w_{G/T_{t-1}},w_{G/T}}$ where $T_i \subsetneq T_{i+1}$ is an ascending filtration of the edge set of $T$. Now define the sector $D_T$ as the union of all open rays between $C_T$ and ``nearby faces at infinity" (this notion is the delicate part of the construction; see \cref{ssec:fibration} for the details, \cref{fig:fiberbundlesunrise} for a first example). 
The map $\pi_T$ is then simply the projection onto $C_T$. 

This defines a smooth fiber bundle that extends to a piecewise smooth fiber bundle on a certain closure of $D_T$ (which includes $C_T$). In fact, it can be extended even further to obtain a fibration $\pi_G$ of $\Delta_G$ over its spine. Furthermore, these maps fit together to form a fibration $\pi$ of the moduli space of graphs over its spine; see \cref{rem:homotopyequivalence}.

With this at hand we are able to rewrite $I_G$ as
\begin{equation*}
   I_G = \sum_{T \subset G} \int_{D_T} \omega_G =  \sum_{T \subset G} \int_{C_T} (\pi_T)_*\omega_G,
\end{equation*}
where $(\pi_T)_*$ denotes the push-forward along $\pi_T$, that is, integration along the fibers of this bundle. 

\begin{figure}[ht]
\begin{tikzpicture}[scale=1]
\coordinate (v0) at (0,0) node {};
  \coordinate  (v1) at (0,2);
   \coordinate  (v2) at (2,2);
   \draw (v1) to[out=90,in=90] (v2);
   \draw[magenta] (v1) -- (v2);
   \draw (v1) to[out=-90,in=-90] (v2);
  \filldraw[fill=black] (v1) circle (0.1);
  \filldraw[fill=black] (v2) circle (0.1);
  \end{tikzpicture}
  \quad  \quad 
 \begin{tikzpicture}[scale=1]
\coordinate (l) at (-3,0);
  \coordinate  (r) at (3,0); 
   \coordinate  (o) at (0,3);
  \coordinate (lo) at (-1.5,1.5); 
   \coordinate (lr) at (0,0); 
  \coordinate (ro) at (1.5,1.5); 
  \coordinate (c) at (0,1.1); 
  \coordinate (t1) at (-1.2,1.42); 
  \coordinate (t4) at (-1,1.36);
    \coordinate (t2) at (-.8,1.31); 
    \coordinate (t5) at (-.58,1.246);
        \coordinate (t3) at (-.37,1.186);
        \coordinate (t6) at (-.17,1.141);
   \draw (l) -- (r);
   \draw (l) -- (o) ;
   \draw (r) -- (o) ;
   \draw[blue] (t1) -- (o);
   \draw[blue] (t1) -- (l);
    \draw[blue] (t2) -- (o);
   \draw[blue] (t2) -- (l);
   \draw[blue] (t3) -- (o);
   \draw[blue] (t3) -- (l);
    \draw[blue] (t4) -- (o);
   \draw[blue] (t4) -- (l);
   \draw[blue] (t5) -- (o);
   \draw[blue] (t5) -- (l);
   \draw[blue] (t6) -- (o);
   \draw[blue] (t6) -- (l);
     \draw[line width=0.5mm, red] (lo) -- (c);
   \draw[line width=0.5mm, red] (lr) -- (c);
   \draw[line width=0.5mm, red] (ro) -- (c);
  \filldraw[fill=red] (lo) circle (0.07);
  \filldraw[fill=red] (lr) circle (0.07);
 \filldraw[fill=red] (ro) circle (0.07);
  \filldraw[fill=red] (c) circle (0.07);
 \filldraw[fill=black] (l) circle (0.07) node[below]{$[1:0:0]$};
 \filldraw[fill=black] (r) circle (0.07) node[below]{$[0:1:0]$};
 \filldraw[fill=black] (o) circle (0.07) node[above]{$[0:0:1]$};
  \end{tikzpicture} 
   \caption{The sunrise graph $G$ and the simplex $\Delta_G \subset \mb P^G=\mb P(\mb C^3)$. The red part is the spine of $\Delta_G$, the simplicial/cubical complex on the central vertex $w_G$ and the three vertices $w_{G/e_i}$. The blue lines indicate the fibers over the cube $C_{T}$ for $T=e_2$, the black vertices, corresponding to graphs $G/\set{e_i,e_j}$, are the faces at infinity.}\label{fig:fiberbundlesunrise}
\end{figure}
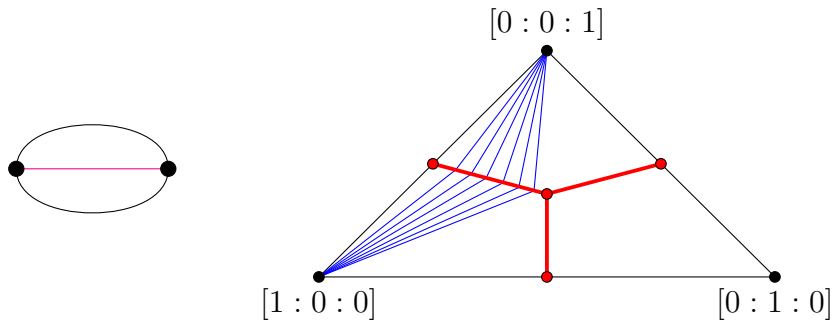

Here it is important to note that both steps are only well-defined if $I_G$ converges absolutely. Thus, to be precise, we should either start with a renormalized integrand or work with some regularization of $\omega_G$ and later worry about the analytic continuation to the points of physical interest. We discuss this in detail in \cref{ssec:parametricFI}.

In summary, we can express a parametric Feynman integral as a linear combination of (simpler) integrals of the push-forwards of $\omega_G$ over the cubes $C_T$. This turns \cref{eq:FIparametric} into
\[
I_G=\sum_{T \subset G} \int_{C_T} \omega_{G,T} \ \text{ with } \ \omega_{G,T}= ({\pi_T})_* \omega_G,
\]
the sum running over all spanning trees of $G$.

This parametric loop-tree duality is arguably much ``cleaner" than its momentum space counterpart. Note, however, that the technical details and intricacies have not simply disappeared; they are hidden in the geometry of the decomposition of $\Delta_G$ and the structure of the fiber bundles $\pi_T$. We demonstrate this by a couple of examples throughout the paper, in particular we use the sunrise diagram (\cref{fig:fiberbundlesunrise}) as a running example. 

When comparing the two variants, it is also important to note that in momentum space we really reduce the number of integration steps; the residue theorem takes care of all $k^0$-integrations. On the other hand, in the parametric version the fiber integration is a non-trivial task, as exemplified by our discussion of the 'wheel with three spokes' in \cref{ssec:examples}.

Potential applications of this construction depend thus on a better understanding of the forms $\omega_{G,T}$, that is, how the combinatorics of the decomposition of $\Delta_G$ and the structure of the graph polynomials $\psi_G,\phi_G$ interact. For instance, in the case of \textit{Feynman periods}, where $\omega_G$ depends only on the Kirchhoff/second Symanzik polynomial $\psi_G$, contraction-deletion or Dodgson identities \cite{fb:periods} may lead to further simplifications. This line of thought also applies to the \textit{canonical forms} on Kontsevich's graph complex defined by Francis Brown in \cite{brown21,brown22}. More comments and a discussion of further applications can be found in \cref{sec:outlook}.

\subsection{Organization of the paper}

We start \cref{sec:geometry} by introducing some notation and conventions. Then we discuss decompositions of the graph simplex $\Delta_G$ and how this gives rise to the fibrations $\pi_T$. An important aspect is understanding the structure of the fibers which takes up most of this section. After this is established, we discuss an example in detail. The section finishes with some comments on the naturality of the construction and other possible fibrations. 

In \cref{sec:fiberintegrals} we quickly recall the definition of and elementary facts about fiber integration.

Section \ref{sec:FI} applies the content of the previous two sections to Feynman integrals in the parametric representation. First, we review Feynman integrals in Schwinger parameters, then we derive the parametric variant of loop-tree duality. In the end we discuss a couple of examples. 

The paper finishes with an outlook in \cref{sec:outlook}.

\section*{Acknowledgements}
This project has received funding by the Royal Society through grant {URF{\textbackslash}R1{\textbackslash}201473} and by the European Research Council (ERC) under the European Union’s Horizon 2020 research and innovation program (grant agreement No.\ 724638). 

The ideas presented here originated from numerous discussions with Karen Vogtmann and Dirk Kreimer about and around Outer space. I also benefited greatly from conversations with Francis Brown and Ralph Kaufmann on the geometry of the spine and its embedding into the moduli space of graphs. Furthermore, I owe many thanks to Erik Panzer for giving me a crash course on his \texttt{Maple} program \texttt{HyperInt} \cite{Panzer_hyperint}.

\section{Simplex fibrations}\label{sec:geometry}

\subsection{Notation and basic definitions}\label{ssec:notation} The following conventions will be used throughout the paper.

\begin{itemize}
    \item We write $x=(x_1,\ldots,x_d)$ for a vector $x \in k^d$ with $k\in \set{\mb R, \mb C}$ and $[x]=[x_1:\ldots:x_d]$ for the corresponding point in $\mb P(k^d)$ -- we sometimes omit the brackets $[ \cdot ]$ if the meaning is clear from the context. We set $\mb R_+=(0,\infty)$.
    \item For a graph $G$ we let $E_G$ and $V_G$ denote its set of (internal) edges and vertices, respectively. We set $e_G=|E_G|$, $v_G=|V_G|$, and we write $h_G=\dim H_1(G;\mb Q)$ for the \textit{rank} or \textit{loop number} of $G$.
    \item All graphs we consider here will be finite and connected. 
    \item A graph is \textit{core} or \textit{1-particle irreducible} (1PI) if removing any edge reduces its first Betti number by one.
    \item A graph may have external edges (\emph{legs})\footnote{They play no role for the geometric arguments of the present section, but become important in the definition of Feynman integrals in \cref{sec:FI}.}. These are defined in the usual way using half-edges: Internal edges are formed by pairs of half-edges, external edges correspond to single non-paired half-edges. A \emph{tadpole} (self-loop) is an edge formed by two half-edges that are connected to the same vertex.
    \item A \emph{tree} is a graph $T$ such that $h_T=0$. A \emph{forest} is a disjoint union of trees.
    \item A \emph{subgraph} $\gamma \subset G$ (without legs) is a graph $\gamma$ such that $V_\gamma\subset V_G$ and $E_\gamma \subset E_G$. If $\gamma\subset G$ is a subgraph, we write $G-\gamma$ for the subgraph of $G$ defined by $V_{G-\gamma}=V_G$ and $E_{G-\gamma}=E_G \setminus E_\gamma$. We write $G/\gamma$ for the graph obtained from $G$ by collapsing each connected component of $\gamma$ to a vertex.
    \item A \emph{spanning tree} of $G$ is a subgraph $T \subset G$ such that $T$ is a tree and $V_T=V_G$. A \emph{spanning forest} of $G$ is a subgraph $F \subset G$ such that $F$ is a forest and $V_F=V_G$.
    \item We sometimes abuse notation by identifying edge sets and subgraphs in the obvious way.
    \item If no confusion is possible, we write $x$ for a singleton $\set x$.
    \item The symbols $ \dot{\in} $ and $ \dot{\subseteq} $ mean ``\ldots pairwise different elements / pairwise disjoint subsets of \ldots ". 
\end{itemize}

\subsection{The graph simplex and its spine}\label{ssec:graphsimplex}

Let $G$ be a graph without tadpoles (the case with tadpoles is discussed in \cref{rem:tadpoles}). Set $\mb P^G=\mb P(\mb C^{e_G})$ with coordinates $x_e$ for $e\in E_G$. Define the (open) \emph{graph simplex} $\Delta_G$ as
\[
\Delta_G= \{ [x] \mid x_e \in \mb R_+ \} \subset \mb P^G.
\]

For every spanning forest $F\subset G$ we define a subset $C_F$ of $\Delta_G$ by  
\begin{equation}\label{eq:cube}
   C_F=\set{ (x_e)_{e\in E_G} \mid  0< x_i < x_j  \text{ for all } i\in E_F, j\in E_{G-F}
\ \wedge \
 x_i=x_j \text{ for all } i,j \in E_{G-F}  } . 
\end{equation}
It can be parametrized by an open cube of dimension $e_F$: To ease notation, suppose that the edges of $G$ are ordered such that the edges of $F$ are labeled by $1,\ldots, f$. Then $C_F$ is parametrized by the map 
\begin{equation}\label{eq:cubeparam}
  \iota_F \colon (0,1)^{f} \longrightarrow \Delta_G, \ (x_1,\ldots,x_f) \longmapsto  \ [x_1 : \ldots : x_f: 1: \ldots:1].  
\end{equation}
See \cref{fig:fiberbundlesunrise,fig:4banana,fig:doublebubble,fig:dunce,fig:fiberbundletadpole} for examples.

The family of sets $\set{C_F \mid F \text{ spanning forest of } G}$ is (the geometric realization of) a cubical complex called the \textit{spine} of $\Delta_G$. More precisely, it is the spine of the corresponding cell in the moduli space of graphs that have rank equal to $h_G$ \cite{cv,hv}.

\subsection{An excursion into the moduli space of graphs and its spine}\label{ssec:modspace}
We sketch the geometric background that motivated the construction of the fibration presented in this paper. This section is not essential and may be skipped on a first read. 

We follow the exposition in \cite{cgp,cgp2} which studies the moduli space of \textit{tropical curves} which are certain weighted, marked metric graphs. A slight adjustment produces the moduli spaces of graphs (as a subset of the latter) which is classically defined as the quotient of Outer space by the action of $\mathrm{Out}(F_n)$ \cite{cv}.\footnote{It is actually much easier to define the moduli space of graphs via the detour through either tropical curves or Outer space. Interested readers are therefore encouraged to consult the given references.} We discuss here the case without legs/external edges, but it easily generalizes to the case where such special (half-)edges are allowed. 

Let $2\leq g\in \mb N$. Define a category $\Gamma_g$ by
\begin{itemize}
    \item $\mathrm{ob}(\Gamma_g)$ is the set of all graphs of rank $g$ that have no bridges/separating edges (1PI) and all vertices at least three-valent.
    \item $\mathrm{hom}(\Gamma_g)$ are given by forest collapses or isomorphisms of graphs. 
\end{itemize}

For each $G \in \Gamma_g$ let 
\[
\sigma(G) = \mb R_{\geq 0}^{E_G}= \set{ \ell \colon E(G) \to \mb R_{\geq 0} \mid \inv{\ell}(0) \text{ is a forest in $G$}        }
\]
denote the space of metrics on $G$ that are allowed to vanish on forests in $G$. Given a morphism $f \in \mathrm{hom}(\Gamma_g)$, define a map $\sigma(f) \colon \sigma\left(f(G)\right) \to \sigma(G) $ by
\[
\ell' \longmapsto \ell \ \text{ where }  \ \ell(e)= \begin{cases} \ell'(e') & \text{ if } e=f(e'), \\
                             0 & \text{ if $f$ collapses $e$.} \end{cases}
\]
This defines a functor from $\Gamma_g$ to the category of topological spaces. The \emph{moduli space of rank $g$ graphs} ${\m {MG}_g}$ is defined as the colimit of this functor. 

If there were no graph isomorphisms in $\mathrm{hom}(\Gamma_g)$, the space $\m {MG}_g$ would be a union of cones, with some of their faces deleted, identified along their common boundaries. Isomorphisms act however non-trivially on this space, so that $\m {MG}_g$ is an orbifold. For instance, if $G$ is the sunrise graph (\cref{fig:fiberbundlesunrise}), then the corresponding cone in $\m {MG}_2$ is the quotient of $\mb R_{\geq 0}^3$ by the action of $S_3$ permuting the coordinates $x_1,x_2,x_3$. 

It is often convenient to normalize the metrics on graphs. For this define a function $\lambda \colon \m {MG}_g \to \mb R_{\geq 0}$ by measuring the \emph{volume} of a metric, $\lambda(\ell)= \sum_{e\in E_G} \ell(e)$. Consider the subset $MG_g \subset \m {MG}_g$ defined by $MG_g =  \inv{ \lambda } (1)$. Note that this turns the cones $\sigma(G)$ into simplices. In the following we will work exclusively with this space and henceforth refer to $MG_g$ as the \emph{moduli space of (rank $g$) graphs}. 

If we ignore isomorphisms for a moment, we can think of $MG_g$ as a semi-simplicial complex (aka $\Delta$-complex \cite{hatcher}) with some of its faces deleted. If $\sigma_G$ is a cell in $MG_g$, parametrizing the space of metrics of unit volume on $G$, then 
\[
\ell \in \partial \sigma_G \ \Longleftrightarrow \ \inv{\ell}(0) \text{ is a non-empty forest in $G$.} 
\]
If on the other hand a metric $\ell$ vanishes on a subgraph $\gamma$ with $h_\gamma>0$, then this face is not in $MG_g$; it is called a \emph{face at infinity}. 

One may compactify $MG_g$ by various methods of which two are important for us. Firstly, we can add all faces at infinity to form a (semi)simplicial completion $\overline{MG}_g$ of $MG_g$ \cite{cgp}. Secondly, we can truncate small ``neigborhoods of infinity" in each cell $\sigma_G$ to obtain polytopes $J_G$ that assemble to a compact space $\ti{MG}_g$, homotopy equivalent to $MG_g$. This is equivalent (homeomorphic) to the iterated blow-up of $\ol{MG}_g$ along the faces at infinity. For the details we refer to \cite{v-bord}; see also \cite{brown21,mb2}.  

In the context of Outer space it is well-known \cite{cv} that $MG_g$ deformation retracts onto its \emph{spine} $S_g$, a subset which has the structure of a cubical complex \cite{hv}. As an abstract complex it is defined as the order complex of the poset $\left( \mathrm{ob}(\Gamma_g), \leq \right)$ where $G \leq G'$ if there is a forest $F \subset G'$ with $G=G'/F$\footnote{Here we are tacitly working with isomorphism classes of graphs}. The geometric realization of this complex can be embedded as a subcomplex $S_g$ of the barycentric subdivision $\beta(\overline{MG}_g)$ of the semi-simplicial completion of $MG_g$. The deformation retract $r: MG_g \to S_g$ is then described by collapsing all the cells that have vertices at infinity, that is, in $\overline{MG}_g \setminus MG_g$.

Above (\cref{ssec:graphsimplex}) we have seen how to embed $S_g$ into $MG_g$. In \cref{rem:homotopyequivalence} we sketch how the maps $\pi_T$ can be assembled to define a fibration $\pi:MG_g \to S_g$, producing another proof that both spaces are homotopy equivalent.

\subsection{A fibration over the spine} \label{ssec:fibration}

Let $T$ be a spanning tree of $G$, that is, a maximal forest $T\subset G$. We want to define a subset $D_T$ of $\Delta_G$ (or $\sigma_G$, as defined in \cref{ssec:modspace}) such that there is a surjective map $\pi_T \colon D_T \to C_T$. Moreover, we want the union of the sets $D_T$ where $T$ ranges over all spanning trees of $G$ to form a partition of $\Delta_G$ (up to a set of measure zero).

Putting the cart before the horse, we start by describing the fibers of this sought-after map. In fact, the fibers are the very thing we are interested in because our eventual goal is to integrate forms along the fibers of this map. In any case, we need to establish that there \emph{is} a map with the asserted properties. 

As mentioned in the introduction, the idea is to define $\pi_T$ by bundling together all fibers that connect $C_T$ to certain subsets of the \emph{faces at infinity} in $\ol \Delta_G$, that is, to the faces where $\set{e\in E_G \mid x_e=0}$ defines a subgraph $\gamma$ with $h_\gamma>0$. 
These subsets of faces will be indexed by certain core subgraphs of $G$. To make this precise consider 
\[
\m C (T) = \set{ \gamma \subsetneq G \mid \gamma \text{ is core,} \ E_\gamma \cap E_T \neq \varnothing, \ E_\gamma \cup E_T \neq E_G  },
\]
 the set of all proper core subgraphs of $G$ that share at least one edge with $T$ and whose union with $T$ does not cover the whole graph. 
This set is partially ordered by the inclusion relation. We let $\mg$ denote its maximal elements.

\begin{lem} \label{lem:cmax}
Let $G$ be a connected graph. The cardinality of $\mg$ is equal to $h_G$, the rank of $G$. In particular, $|\mg|$ is independent of the choice of spanning tree $T$ of $G$.
\end{lem}

\begin{proof}
A core subgraph $\gamma$ of $G$ is maximal if and only if $h_\gamma=h_G-1$.
A simple model for $\mg$ is obtained by collapsing all edges in $T$. Then $G/T$ is a rose on $h_G$ petals ($h_G$ tadpoles/self-loops connected to a single vertex). There is a bijection 
\[
\mg \cong \set{I\subset E_{G/T} \colon |I|=h_{G}-1}.
\]
The injection $\hookleftarrow$ is given by expanding $T$ again and chopping off edges from the subgraph $I\cup T$ to make it core. For the other direction $\hookrightarrow$ recall that every $\gamma$ in $\mg$ shares at least one edge with $T$ and its union with $T$ is not the full graph $G$. Thus, after collapsing $T$ the subgraph $\gamma$ becomes a rose with at most $h_G-1$ petals. If the number of petals is less than $h_G-1$, then $\gamma $ was not maximal.
\end{proof}

We now describe the fibers of the map $\pi_T$. Given a point $x \in C_T$ we let the fiber of $\pi_T$ over $x$ be the set
\begin{equation}\label{eq:cone}
 \inv{\pi_T}(x)=\bigcup_{\gamma_1, \ldots, \gamma_{h_G-1} \din \mg} \text{Cone}_x\left(v_{\gamma_1}(x), \ldots, {v_\gamma}_{h_G-1}(x)\right)   
\end{equation}
where $v_\gamma$ is the map that sends all $x_e$ with $e\in E_\gamma$ to zero,
\[
v_\gamma \colon C_T \longrightarrow \ol \Delta_{G/\gamma} \subset \ol \Delta_G , \quad x_e \longmapsto 
\begin{cases}
x_e & \text{ if } e \notin E_\gamma, \\
0 & \text{ else, }
\end{cases}
\]
and $\text{Cone}_x\left(y_1, \ldots, y_n \right)$ denotes the open convex cone in $\Delta_G$ that is based at $x$ and spanned by the $y_i$, 
\[
\text{Cone}_x\left(y_1, \ldots, y_n \right)=\big\{ [  x + \mu_1 y_1 + \ldots + \mu_n y_n ] \mid  \mu_i \in \mb R_+  \big\}
.
\]

\begin{example}
Consider the icecream-cone graph, depicted in \cref{fig:dunce}, with spanning tree $T=\set{e_2,e_4}$. Then $\mg= \set{\gamma_{124},\gamma_{34}}$ where $\gamma_I=\set{e_i \mid i\in I }$. The cube $C_T$ is parametrized as in \cref{eq:cubeparam},
\[
C_T=\set{ [1:x_2:1:x_4] \mid x_2,x_4 \in (0,1) },
\]
and the maps $v_{\gamma_{124}}$ and $v_{\gamma_{34}}$ are given by
\[
v_{\gamma_{124}}\colon [1:x_2:1:x_4] \longmapsto [0:0:1:0],
\quad
v_{\gamma_{34}} \colon [1:x_2:1:x_4]\longmapsto [1:x_2:0:0].
\]
We see that $v_{\gamma_{124}}$ maps the whole cube into the vertex at infinity, while $v_{\gamma_{34}}$ is non-constant; its image covers half of the line $\{x_3=x_4=0\}$ at infinity.
\end{example}

\begin{figure}[ht]
  \begin{tikzpicture}[scale=.8]
  \coordinate (h) at (0,0) node {};
   \coordinate (v0) at (0,2.5);
   \coordinate  (v1) at (2,1.5);
   \coordinate (v2) at (2,3.5);
   \draw[magenta] (v0) -- (v1) node [midway,below]{$\color{black}e_2$};
   \draw (v0) -- (v2) node [midway,above]{$e_1$};
   \draw (v2) to[out=-135,in=135] (v1) node [xshift=-.55cm,yshift=.75cm]{$e_3$};
   \draw[magenta] (v2) to[out=-45,in=45] (v1) node [xshift=.55cm,yshift=.75cm]{$\color{black}e_4$}; 
   \fill[black] (v0) circle (.1cm);
   \fill[black] (v1) circle (.1cm); 
   \fill[black] (v2) circle (.1cm); 
  \end{tikzpicture}
  \quad \quad 
\begin{tikzpicture}[scale=2.3]
\coordinate (v1) at (-1,0);
\coordinate  (v2) at (1,-.2); 
\coordinate  (v4) at (0.15,1.44); 
\coordinate  (v3) at (1.23,.5);
 \coordinate (c2) at (.44,.1); 
\coordinate (c3) at (0.13,.7); 
  \coordinate (c4) at (0.18,0.26);   
 \coordinate (c1) at (.4,0.5);
  \coordinate (h34) at (1.23,.5);
  \coordinate (h12) at (-.05,-.095);
 \fill[fill=cyan!50] (v1) -- (c1) -- (c2) -- (h12) -- (v1);
 \fill[fill=cyan!50] (c3) -- (h34) -- (v3) -- (c4) -- (c3);
  \fill[fill=cyan!50] (v1) -- (c3) -- (c4) -- (v1);
  \fill[fill=cyan!50] (h12) -- (c1) -- (c2) -- (h12);
    \fill[fill=cyan!50] (c1) -- (c3) -- (h34) -- (c1);
    \fill[fill=cyan!50] (c4) -- (c2) -- (v3) -- (c4);
  \draw[blue] (c3) to (v1);
  \draw[blue, dashed] (c4) to (v1);
  \draw[blue] (c2) to (h12);
  \draw[blue] (c3) to (h34);
  \draw[blue, dashed] (c4) to (v3);
  \draw[blue] (c2) to (v3);
 \filldraw[fill=black] (v1) circle (0.03) node[left]{$1$};
  \filldraw[fill=black] (v2) circle (0.03) node[right]{$2$};
   \filldraw[fill=black] (v4) circle (0.03) node[above]{$4$};
    \filldraw[fill=black] (v3) circle (0.03) node[right]{$3$};
  \draw (v1) to (v2);
  \draw (v1) to (v4);
  \draw (v2) to (v4);
  \draw (v4) to (v3);
    \draw (v2) to (v3);
         \fill[fill=red] (c1) -- (c2) -- (c4) -- (c3) -- (c1);
    \draw (c1) to (c2);
    \draw[dashed] (c2) to (c4);
  \draw[dashed] (c4) to (c3);
  \draw (c3) to (c1);
      \filldraw[fill=blue] (h12) circle (0.023);
   \draw[blue] (c1) to (h34);
   \filldraw[fill=red] (c2) circle (0.023);
   \filldraw[fill=red] (c3) circle (0.023);
     \filldraw[fill=red] (c4) circle (0.023); 
     \draw[blue] (c1) to (h12); 
     \filldraw[fill=orange] (c1) circle (0.033);
      \draw (c1) to (c2);
  \end{tikzpicture}
\caption{Cube $C_T$ (red) and sector $D_T$ (blue) for $T=\set{e_2,e_4}$ in the icecream-cone graph. The orange vertex marks the center of $\Delta_G$. Here $\mg$ consists of the two subgraphs with edge sets $\set{e_1,e_2,e_4}$ and $\set{e_3,e_4}$.}\label{fig:dunce}
\end{figure}
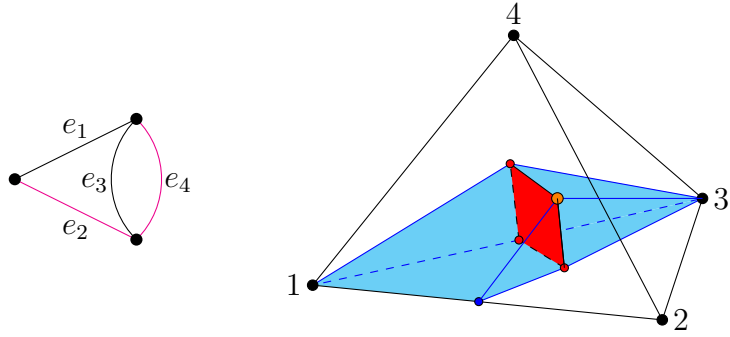

See \cref{fig:fiberbundlesunrise,fig:4banana,fig:doublebubble,fig:fiberbundletadpole} for further examples.
\newline

In order to define the map $\pi_T$ we need the following

\begin{lem}\label{lem:welldefined}
For $x\neq x' \in C_T$ we have $\inv{\pi_T}(x) \cap \inv{\pi_T}(x') = \varnothing.$
\end{lem}

\begin{proof}
First of all note that $\text{Cone}_x\left(v_{\gamma_1}(x), \ldots, {v_\gamma}_{h_G-1}(x)\right)$ is an open cone, it does not include its boundary faces, spanned by the rays from $x$ to some of the $v_\gamma(x)$. It also does not meet any face at infinity.
Hence, it suffices to consider a fixed cone $C_x= \text{Cone}_x\left(v_{\gamma_1}(x), \ldots, {v_\gamma}_{h_G-1}(x)\right)$ and vary $x$.

So let $x\neq x' \in C_T$ and suppose first that $h_G=2$. This means $C_x= \text{Cone}_x\left(v_{\gamma}(x)\right)$ for some $\gamma$ in $\mg$. If $v_\gamma(x)= v_\gamma(x')$, then $C_x \cap C_{x'}=\varnothing$. If not, then using the parametrization by \cref{eq:cubeparam} we can find a non-empty subset $I\subset E_T\setminus  E_\gamma$ such that 
\[
x_i\neq x_i'  \ \Longleftrightarrow \ i\in I. 
\]
Points in the intersection of $C_x$ and $C_{x'}$ satisfy
\[
x + \mu v_\gamma(x)  = \lambda (  x' + \mu' v_\gamma(x') )  \quad (\lambda >0).
\]
Explicitly we get
\begin{enumerate}
\item $x_e= \lambda x_e'$ for $e\in E_\gamma\cap E_T$,
\item $1 = \lambda 1$ for $e\in E_\gamma \setminus E_T$,
\item $x_e + \mu x_e = \lambda ( x_e' + \mu' x_e')$ for $e\in E_T \setminus E_\gamma$,
     and
    \item $1+ \mu 1= \lambda (1+\mu'1)$ for $e\in E_G \setminus (E_\gamma \cup E_T)$.
\end{enumerate}
Since $I$ is disjoint from $E_\gamma$, the first item reads $x_e= \lambda x_e$. Thus, $\lambda=1$. With (4) we find then $\mu=\mu'$. Finally, from (2) and (3) we infer $x_e = x_e'$ for all $e\in E_G$, a contradiction.  

The general case $h_G > 2$ follows by the same reasoning as above, using the following claim: For any subset $A \subset \mg$ with $|A|=h_G-1$
\begin{enumerate}
    \item there is an edge $e$ that does not belong to $T$, but to every element of $A$.
    \item for any $\gamma \in A$ there is an $e(\gamma)$ in the complement of $T$ and $\gamma$ which belongs to any other $\gamma' \in A$.
\end{enumerate}  

This is a corollary of \cref{lem:cmax}. We can think of an element $\gamma \in \mg$ as a subset $I\subset \set{1,\ldots, h_G}$ of cardinality $h_G-1$, or equivalently, as a one-element subset of the latter. That is, we can specify an element $\gamma \in \mg$ uniquely by an edge in the complement of $T$ that does not belong to $\gamma$, but to any other $\gamma' \in \mg$. By the same argument, $A$ is determined by choosing $h_G-1$ ``forbidden edges" out of $\set{1,\ldots,h_G}$. Thus, there is always one left which has to belong to every $\gamma \in A$.

With this we can solve the system of equations for $C_x \cap C_{x'}$. (1) leads to $\lambda=\lambda'$ and (2) leads to $\mu_i=\mu_i'$ for every $i=1,\ldots, h_{G-1}$. Hence, $x=x'$, a contradiction.
\end{proof}

\begin{figure}[ht]
\begin{tikzpicture}[scale=1]
  \coordinate  (v0) at (0,0) node {}; 
    \coordinate  (v1) at (0,2); 
   \coordinate  (v2) at (1.3,2);
   \coordinate (v3) at (2.6,2); 
   \draw (v1) to [out=90,in=90] node[above] {$e_1$} (v2);
   \draw (v1) to [out=-90,in=-90] node[below] {$e_2$} (v2);
   \draw (v2) to [out=90,in=90] node[above] {$e_3$} (v3);
   \draw (v2) to [out=-90,in=-90] node[below] {$e_4$} (v3);
  \filldraw[fill=black] (v1) circle (0.1);
  \filldraw[fill=black] (v2) circle (0.1);
    \filldraw[fill=black] (v3) circle (0.1);
  \end{tikzpicture}
  \ 
\begin{tikzpicture}[scale=2.2]
\coordinate (v1) at (-1,0);
\coordinate  (v2) at (1,-.2); 
\coordinate  (v4) at (0.15,1.44); 
\coordinate  (v3) at (1.23,.5);
 \coordinate (c1) at (-0.4,0.75);
 \coordinate (c2) at (-.01,.5); 
\coordinate (c3) at (0.13,.75); 
  \coordinate (c4) at (.4,0.6);  
  \coordinate (h34) at (.65,1);
  \coordinate (h12) at (-.1,-.09);
 \draw (c1) to (c2);
  \draw (c2) to (c4);
  \draw (c4) to (c3);
  \draw (c3) to (c1);
 \fill[fill=cyan!50] (v1) -- (h12) -- (c3) -- (v1);
  \fill[fill=cyan!50] (c1) -- (v4) -- (h34) -- (c3) -- (c1);
  \fill[fill=cyan!50] (v1) -- (c1) -- (c3) -- (v1);
   \fill[fill=cyan!50] (h12) -- (c2) -- (c4) -- (h12);
    \fill[fill=cyan!50] (c1) -- (c2) -- (v4) -- (c1);
  \fill[fill=cyan!50] (c4) -- (c3) -- (h34) -- (c4);
  \draw[blue] (c1) to (v4);
      \draw[blue] (c4) to (h34);
       \draw[blue] (c3) to (h34);
   \draw[blue] (c4) to (h12);
    \draw[blue] (c2) to (h12);
     \draw[blue] (c3) to (v1);
      \draw[blue] (c1) to (v1);
 \filldraw[fill=black] (v1) circle (0.03) node[left]{$1$};
  \filldraw[fill=black] (v2) circle (0.03) node[right]{$2$};
   \filldraw[fill=black] (v4) circle (0.03) node[above]{$4$};
    \filldraw[fill=black] (v3) circle (0.03) node[right]{$3$};
  \draw (v1) to (v2);
  \draw (v1) to (v4);
  \draw (v2) to (v4);
  \draw (v4) to (v3);
  \draw[dashed] (v1) to (v3);
    \draw (v2) to (v3);
         \fill[fill=red] (c1) -- (c2) -- (c4) -- (c3) -- (c1);
     \draw[dashed] (c1) to (c2);
  \draw (c2) to (c4);
  \draw (c4) to (c3);
  \draw[dashed] (c3) to (c1);
   \draw[blue] (c2) to (v4);
   \filldraw[fill=red] (c1) circle (0.023);
   \filldraw[fill=red] (c2) circle (0.023);
   \filldraw[fill=red] (c3) circle (0.023);
     \filldraw[fill=orange] (c4) circle (0.033);
      \filldraw[fill=blue] (h12) circle (0.023);
   \filldraw[fill=blue] (h34) circle (0.023);
  \end{tikzpicture}
   \quad 
\begin{tikzpicture}[scale=2.2]
\coordinate (v1) at (-1,0);
\coordinate  (v2) at (1,-.2); 
\coordinate  (v4) at (0.15,1.44); 
\coordinate  (v3) at (1.23,.5);
 \coordinate (c2) at (.44,.1); 
\coordinate (c3) at (0.13,.7); 
  \coordinate (c4) at (0.18,0.26);   
 \coordinate (c1) at (.4,0.6);
  \coordinate (h34) at (.65,1);
  \coordinate (h12) at (-.09,-.09);
 \fill[fill=cyan!50] (v1) -- (c1) -- (c2) -- (h12) -- (v1);
 \fill[fill=cyan!50] (c3) -- (h34) -- (v3) -- (c4) -- (c3);
  \fill[fill=cyan!50] (v1) -- (c3) -- (c4) -- (v1);
  \fill[fill=cyan!50] (h12) -- (c1) -- (c2) -- (h12);
    \fill[fill=cyan!50] (c1) -- (c3) -- (h34) -- (c1);
    \fill[fill=cyan!50] (c4) -- (c2) -- (v3) -- (c4);
  \draw[blue] (c3) to (v1);
  \draw[blue, dashed] (c4) to (v1);
  \draw[blue] (c2) to (h12);
  \draw[blue] (c3) to (h34);
  \draw[blue, dashed] (c4) to (v3);
  \draw[blue] (c2) to (v3);
 \filldraw[fill=black] (v1) circle (0.03) node[left]{$1$};
  \filldraw[fill=black] (v2) circle (0.03) node[right]{$2$};
   \filldraw[fill=black] (v4) circle (0.03) node[above]{$4$};
    \filldraw[fill=black] (v3) circle (0.03) node[right]{$3$};
  \draw (v1) to (v2);
  \draw (v1) to (v4);
  \draw (v2) to (v4);
  \draw (v4) to (v3);
    \draw (v2) to (v3);
         \fill[fill=red] (c1) -- (c2) -- (c4) -- (c3) -- (c1);
     \draw (c1) to (c2);
  \draw[dashed] (c2) to (c4);
  \draw[dashed] (c4) to (c3);
  \draw (c3) to (c1); 
  \draw[blue] (c1) to (h34);
     \filldraw[fill=red] (c2) circle (0.023);
   \filldraw[fill=red] (c3) circle (0.023);
     \filldraw[fill=red] (c4) circle (0.023);
      \filldraw[fill=blue] (h12) circle (0.023);
   \filldraw[fill=blue] (h34) circle (0.023);
   \draw[blue] (c1) to (h12);
   \filldraw[fill=orange] (c1) circle (0.033); 
  \end{tikzpicture}
\caption{Cubes $C_T$ (red) and sectors $D_T$ (blue) for $T=\set{e_2,e_3}$ and $\set{e_2,e_4}$. The orange vertex marks the center of $\Delta_G$. The two blue points are the vertices of $\beta(\ol \Delta_G)$ that are represented by the two core subgraphs on edges $e_3,e_4$ and $e_1,e_2$, respectively.}\label{fig:doublebubble}
\end{figure}
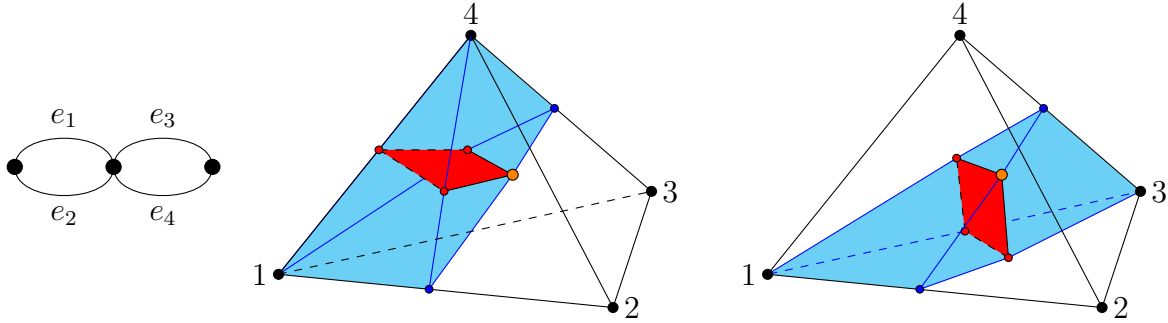

The previous lemma allows to finally define the map $\pi_T$. We set $D_T= \bigcup_{x\in C_T} \inv{\pi_T}(x)$ and define $\pi_T$ by 
\be \label{eq:defnpi}
\pi_T(y)=x  \  \Leftrightarrow \  y \in  \text{Cone}_x\left(v_{\gamma_1}(x), \ldots, {v_\gamma}_{h_G-1}(x)\right)  \text{ for } \gamma_1,\ldots,\gamma_{h_G-1} \din \mg  .
\ee
This is well-defined by \cref{lem:welldefined}.
\newline

The closure of $D_T$ in $\ol \Delta_G$ has a convenient description in terms of $\beta(\ol \Delta_G)$, the first barycentric subdivision of $\ol \Delta_G$\footnote{This point of view is very much inspired by \cite{berger-kaufmann:WBK} which discusses the passage from simplices to cubes in a much more general setting.}. The vertices of the latter complex are in one-to-one correspondence with proper subgraphs $\gamma \subsetneq G$. We can parametrize them by
\[
w_\gamma= [ x_e ] \text{ with }  x_e=\begin{cases}0 & \text{ if } e \in E_{\gamma} \\ 1 & \text{ else.}
\end{cases}
\]
Then $D_T$ is the interior of the subcomplex of $\beta(\ol \Delta_G)$ spanned by the vertex set 
\[
\set{w_\gamma \mid  E_\gamma \cap E_T \neq \varnothing, \ E_\gamma \cup E_T \neq E_G  } \cup \set{w_F \mid F\subset T}.  
\]
By definition $\ol D_T$ contains all rays from points in $C_T$ (including the vertices of type $w_F$) to the vertices of type $w_\gamma$ with $\gamma \in \mg$. Furthermore, if $\gamma_1\cap \gamma_2 \neq \varnothing$, then the closure of $\mathrm{Cone}_x(v_{\gamma_1}(x),v_{\gamma_2}(x))$ contains the vertex $w_{\gamma_1 \cap \gamma_2}$:
\be\label{eq:intersectionvertex}
w_{\gamma_1 \cap \gamma_2} = \lim_{\mu \to \infty}  [  x + \mu v_{\gamma_1}(x) + \mu v_{\gamma_2}(x) ] \  \text{ for any }  x    \text{ with }  x_e=1  \text{ if }  e \in E_{\gamma_1 \cap \gamma_2}.
  \ee 
  Additionally, $\ol D_T$ contains all the vertices $w_{\gamma \cup F}$ for $F\subset T$ and $\gamma \in \mg$:
\be \label{eq:noncorevertex}
w_{\gamma \cup F} = v_\gamma(x) \text{ for } x_e= \begin{cases}0 & \text{ if } e \in E_{F} \\ 1 & \text{ else.}\end{cases}
\ee 
Both \cref{eq:intersectionvertex,eq:noncorevertex} have solutions $x$ in (the closure of) $C_T$.

\begin{prop}\label{prop:covering}
Up to a set of measure zero (codimension greater than one) the sets $D_T$ partition the simplex $\Delta_G$,
 \[
 \Delta_G \subset \bigcup_{T \in \m T(G)} \overline D_T   \quad \text{ and }  \ \bigcap _{T \in \m T(G)} {D_T} = \varnothing.
 \]
\end{prop}

\begin{proof}

The first assertion is true because every point $x$ in $\Delta_G$ lies in a maximal simplex $\sigma=(w_\varnothing,w_{e_1},w_{\set{e_1,e_2}},\ldots)$ of $\beta(\ol \Delta_G)$. Let $F$ be the forest in $G$, defined by 
\[
E_F=\max\set{I \subset E_G \mid w_I \in \sigma} \cap \set{I \subset E_G \mid I \text{ is a forest in $G$}}.
\]
Then we can use \cref{eq:intersectionvertex,eq:noncorevertex} to infer that $x \in \ol D_T$ for any tree $T\subset G$ such that $F\subset T$.

The second assertion follows from the fact that $\ol D_T$ and $\ol D_{T'}$ intersect in the closure of $ D_F=\bigcup_{x\in C_F} \inv{\pi_T}(x)$ for $F=T\cap T'$. But $D_F$ is contained in the boundary of both $\ol D_T$ and $\ol D_{T'}$, hence $ D_T \cap D_{T'}=\varnothing$.
\end{proof}

\begin{rem}
The examples in \cref{fig:fiberbundlesunrise,fig:dunce,fig:doublebubble} show that this partition of $\Delta_G$ is different from the ``classical" one using Hepp sectors \cite{Hepp66}. 
\end{rem}

\begin{prop}\label{prop:fiberbundle}
  The map $ \pi_T \colon D_T  \to  C_T$, defined in \cref{eq:defnpi}, is a smooth fiber bundle whose fibers are diffeomorphic to $h_G$ copies of $\mb R^{e_G-e_T-1}=\mb R^{h_G-1}$. 
\end{prop}

\begin{proof}
Recall that $D_T$ is a disjoint union of $h_G$ open cones over the points in $C_T$. The map $\pi_T$ simply maps all points in a cone to its basepoint which is well-defined by \cref{lem:welldefined}, and globally trivial (by construction, and also because the base is contractible). Moreover, it is obviously smooth. 
\end{proof}

Note that $\pi_T$ is well-defined on $\ol D_T$ minus the points at infinity in $\ol \Delta_G$. We could thus extend the family $\{\pi_T\colon D_T  \to  C_T\}_{T\in \m T(G)}$ to a single continuous -- in fact, piecewise smooth -- map $\pi_G$ on $\Delta_G$ (but not on $\ol \Delta_G$). In other words, it extends to $\ti{\Delta}_G$, the closure of $\Delta_G$ in the compactification $\ti{MG}_g$ (Figure \ref{fig:sunrisecompact} indicates the map $\pi_G$ on the compactified cell $\ti \Delta_G$).

In both cases the resulting map is no longer a fiber bundle, but only a fibration\footnote{Fibers are not necessarily homeomorphic, but have the same homotopy type (\cite[\S 4]{hatcher}).}. The fibers over cubes $C_F$ for $F$ a spanning $k$-forest with $k>1$ are unions of faces of the (closure of the) cones \eqref{eq:cone}, hence also contractible, but not necessarily homeomorphic to the other fibers.

For an example see \cref{fig:fiberbundlesunrise} (or \cref{fig:sunrisecompact}): 
None of the three maps $\pi_T$ is defined at the three corners of $\ol \Delta_G$. Each map can be extended to the lines $\set{x_i=x_j,x_k>x_i \colon i,j,k \din  \set{1,2,3}}$, but the fiber over the center $w_\varnothing=[1:1:1]$ is not homeomorphic to $\mb R$. 

At this point one might wonder why we do not work with $\ti \Delta_G$ instead of $\Delta_G$. This has certain advantages, for instance making the fibers of each $\pi_T$ compact. This in turn would streamline the application of fiber integration in \cref{ssec:parametricLTD}. The price to pay is that the map $\pi_T$, in particular its fibers, are more difficult to describe in this setting.

\begin{figure}[ht]
     \begin{tikzpicture}[scale=1]
\coordinate (l) at (-2.5,0);
  \coordinate  (r) at (2.5,0); 
   \coordinate  (o) at (0,3);
  \coordinate (lo) at (-1.25,1.5); 
   \coordinate (lr) at (0,0); 
  \coordinate (ro) at (1.25,1.5); 
  \coordinate (c) at (0,1.1); 
  \coordinate (co1) at (-0.33,2.6); 
    \coordinate (co2) at (0.33,2.6); 
 \coordinate (lo1) at (-2.166,0.4); 
  \coordinate (lo2) at (-1.866,0); 
   \coordinate (ro1) at (2.166,0.4); 
  \coordinate (ro2) at (1.866,0); 
  \coordinate (t1) at (-1.2,1.42); 
  \coordinate (t4) at (-1,1.36);
    \coordinate (t2) at (-.8,1.31); 
    \coordinate (t5) at (-.58,1.246);
        \coordinate (t3) at (-.37,1.186);
        \coordinate (t6) at (-.17,1.141);
  \draw[thick] (co1) -- (co2) -- (ro1) -- (ro2) -- (lo2) -- (lo1) -- (co1);
   \draw[blue] (t1) -- (o);
   \draw[blue] (t1) -- (l);
    \draw[blue] (t2) -- (o);
   \draw[blue] (t2) -- (l);
   \draw[blue] (t3) -- (o);
   \draw[blue] (t3) -- (l);
    \draw[blue] (t4) -- (o);
   \draw[blue] (t4) -- (l);
   \draw[blue] (t5) -- (o);
   \draw[blue] (t5) -- (l);
   \draw[blue] (t6) -- (o);
   \draw[blue] (t6) -- (l);
   \draw[cyan] (o) -- (c);
    \draw[cyan] (l) -- (c);
     \draw[cyan] (r) -- (c);
     \draw[line width=0.5mm, red] (lo) -- (c);
   \draw[line width=0.5mm, red] (lr) -- (c);
   \draw[line width=0.5mm, red] (ro) -- (c);
  \filldraw[fill=red] (lo) circle (0.07);
  \filldraw[fill=red] (lr) circle (0.07);
 \filldraw[fill=red] (ro) circle (0.07);
  \filldraw[fill=red] (c) circle (0.07);
\filldraw[draw=white,fill=white] (co1) -- (co2) -- (o) -- (co1);
\filldraw[draw=white,fill=white] (lo1) -- (lo2) -- (l) -- (lo1);
\filldraw[draw=white,fill=white] (ro1) -- (ro2) -- (r) -- (ro1);
\draw[thick] (co2) -- (co1);
\draw[thick] (lo2) -- (lo1);
\draw[thick] (ro2) -- (ro1);
   \draw[dashed] (l) -- (lo1);
    \draw[dashed] (l) -- (lo2);
    \draw[dashed] (o) -- (co1);
    \draw[dashed] (o) -- (co2);
    \draw[dashed] (r) -- (ro1);
    \draw[dashed] (r) -- (ro2);
  \end{tikzpicture} 
    \caption{The compactified cell $\ti \Delta_G$ for the sunrise graph. The blue lines indicate the fibers of $\pi_G : \ti \Delta_G \to S_G$ over the cube $C_{e_2}$, the cyan tripod is the fiber over the central 0-cube $C_\varnothing=[1:1:1]$.}
    \label{fig:sunrisecompact}
\end{figure}
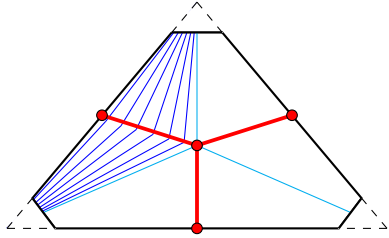

\begin{rem} \label{rem:homotopyequivalence}
The fact that there is a fibration $\pi_G$ from $\ol \Delta_G$ minus its faces at infinity to $\bigcup_{F\subset G} C_F$, where the union is over all spanning forests of $G$, implies that the two spaces are homotopy equivalent. Moreover, the construction of $\pi_G$ is functorial with respect to forest collapses and isomorphisms: 
\begin{enumerate}
\item If $G$ and $G'$ are isomorphic as graphs, then $\pi_G$ and $\pi_{G'}$ are isomorphic as bundles.
\item Each $\pi_T$ extends to the faces $\set{x_i=0 \mid i \in I, I\subset E_T}$ by $\pi_{T/I}=(\pi_{T})_{|\set{x_i=0\mid i \in I} }$. If there is $I \subset T,T'$ for two trees $T\subset G$ and $T'\subset G'$, such that $G/I=G'/I$, then by (1) the respective bundles $\pi_T$ and $\pi_{T'}$ agree when restricted to the corresponding common face of $\ol \Delta_G$ and $\ol \Delta_{G'}$.
 \end{enumerate}
Thus, we can combine the above observations to deduce the existence of a homotopy equivalence $\pi$ between the moduli space of graphs and its spine, a well-known fact \cite{cv}. See also \cite{berger-kaufmann:WBK} for another proof in a similar setting but without providing an explicit map.
\end{rem}

We finish this section with an example.

\begin{example}\label{eg:sunrise}
Consider the sunrise diagram depicted in \cref{fig:fiberbundlesunrise}. If we let $T=e_1$, then ${C}_T=\set{[x:y:y] \mid 0<x<y }$. The set $\mg$ consists of the two subgraphs on edges $e_1,e_2$ and $e_1,e_3$, and the respective $v_\gamma$ maps $C_T$ into $[0:0:y]$ and $[0:y:0]$. We find thus 
\[
D_T=\set{[x:y_1:y_2] \mid 0<x<y_i ,\ i=1,2 }
\]
and
\[
\pi_T \left([x:y_1:y_2]\right) = [x:\min(y_1,y_2):\min(y_1,y_2)]  . 
\]

If we let $\varphi$ denote the affine chart $(x_1:x_2:x_3)\mapsto (\frac{x_1}{x_3},\frac{x_2}{x_3})$ we have in these coordinates
\[
\varphi(C_T) = \set{ (z,1) \mid 0<z<1 }, \quad \varphi(D_T)=\set{ (z,w) \in (0,1)\times \mb R_+ \mid 0<z<\min(w,1)}
\]
so that $\pi_T\colon D_T\to C_T$ is locally given by the map $\pi_T^{loc} = \inv{\varphi}\circ \pi_T \circ \varphi$, 
\[
\pi_T^{loc} \colon \varphi(D_T) \longrightarrow \varphi(C_T), \  (z,w) \longmapsto  \left(\frac{z}{\min(w,1)},1\right)= \begin{cases} (z,1) & \text{ if } w>1 ,\\
(\frac{z}{w},1) & \text{ if } w<1. \end{cases}
\]
Note that $z<w$, so it really maps into the unit cube. The fiber over a point $(u,1)$ is
\[
\inv{\pi_T}(u,1)= (\set{u} \times [1,\infty]) \cup \set{ (uv,v) \mid  v \in (0,1)} \subset D_T.
\]

By the obvious symmetry of $G$ we have an analogous description for the other two spanning trees of $G$.

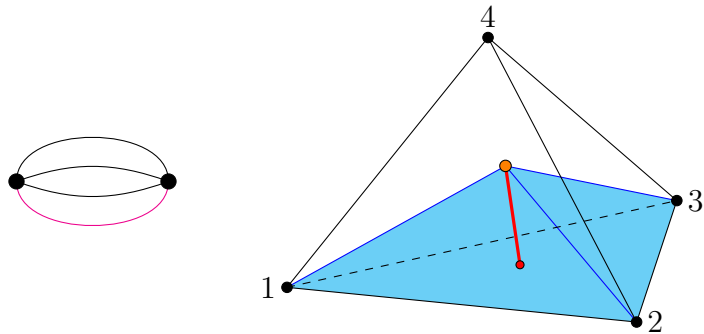
\begin{figure}[ht]
\begin{tikzpicture}[scale=1]
\coordinate (v0) at (0,0) node {};
  \coordinate  (v1) at (0,2);
   \coordinate  (v2) at (2,2);
   \draw (v1) to[out=90,in=90] (v2);
   \draw (v1) to[out=20,in=160] (v2);
   \draw (v1) to[out=-20,in=-160] (v2);
   \draw[magenta] (v1) to[out=-90,in=-90] (v2);
  \filldraw[fill=black] (v1) circle (0.1);
  \filldraw[fill=black] (v2) circle (0.1);
  \end{tikzpicture}
  \quad \quad 
\begin{tikzpicture}[scale=2.3]
\coordinate (v1) at (-1,0);
\coordinate  (v2) at (1,-.2); 
\coordinate  (v4) at (0.15,1.44); 
\coordinate  (v3) at (1.23,.5);
 \coordinate (c1) at (0.25,0.7);
 \coordinate (c2) at (0.333,.13); 
\coordinate (x) at (0.05,.75); 
 \draw (c1) to (c2);
 \fill[fill=cyan!50] (v1) -- (v2) -- (c1) -- (v1);
 \fill[fill=cyan!50] (v2) -- (v3) -- (c1) -- (v2);
  \draw[blue] (c1) to (v1);
  \draw[blue] (c1) to (v2);
  \draw[blue] (c1) to (v3);
 \filldraw[fill=black] (v1) circle (0.03) node[left]{$1$};
  \filldraw[fill=black] (v2) circle (0.03) node[right]{$2$};
   \filldraw[fill=black] (v4) circle (0.03) node[above]{$4$};
    \filldraw[fill=black] (v3) circle (0.03) node[right]{$3$};
  \draw (v1) to (v2);
  \draw (v1) to (v4);
  \draw (v2) to (v4);
  \draw (v4) to (v3);
   \draw[dashed] (v1) to (v3);
    \draw (v2) to (v3);
     \draw[red, very thick] (c1) to (c2);
      \filldraw[fill=red] (c2) circle (0.023);
   \filldraw[fill=orange] (c1) circle (0.033);
  \end{tikzpicture}
\caption{Cube $C_T$ and sector $D_T$ for $T=e_4$ in the 4-edge-banana graph $B_4$. Note how the fiber over $[1:1:1:0]\in \ol C_T$, consisting of three open cones, equals the union of the three sectors for $B_3=B_4/e_4$ in \cref{fig:fiberbundlesunrise}.}\label{fig:4banana}
\end{figure}
\end{example}

Note that this example generalizes in a straightforward way to the ``$n$-edge banana" graphs, consisting of two vertices connected by $n$ edges. For each $n$ the spine is a one-dimensional ``star", formed by $n$ rays connecting the center of $\Delta_G\subset \mb P(\mb C^n)$ to the barycenter of its $n$ codimension one facets. The fiber over each point in a 1-cube $C_T$ is the union of $h_G=n-1$ cones, each homeomorphic to $\mb R^{n-2}$ (the closure PL homeomorphic to a simplex of dimension ${n-2}$). See \cref{fig:4banana}.

\begin{rem}\label{rem:tadpoles}
If the graph $G$ has tadpoles the construction has to be adjusted. Consider, for instance, a graph $G$ on three edges formed by a 2-edge banana with a tadpole (\cref{fig:fiberbundletadpole}). Then it is easy to see that the definitions in \cref{eq:cone,eq:defnpi} are not sufficient; the union of the fibers $\inv{\pi_T}(x)$ does not cover $\Delta_G$. This case can be repaired by adding the tadpole $e_3$ to $\mg$, but it is unclear how to do this in general.
\end{rem}

\begin{figure}[ht]
\begin{tikzpicture}[scale=1]
  \coordinate  (v0) at (0,0) node {}; 
    \coordinate  (v1) at (0,2); 
   \coordinate  (v2) at (1.3,2);
   \draw[magenta] (v1) to [out=90,in=90] node[above] {\color{black}$e_1$} (v2);
   \draw[scale=3] (v2) to [out=60,in=300,loop] node[right] {$e_3$} (v2);
   \draw (v1) to [out=-90,in=-90] node[below] {$e_2$} (v2);
  \filldraw[fill=black] (v1) circle (0.1);
  \filldraw[fill=black] (v2) circle (0.1);
  \end{tikzpicture}
  \quad 
 \begin{tikzpicture}[scale=1]
\coordinate (l) at (-3,0);
  \coordinate  (r) at (3,0); 
   \coordinate  (o) at (0,3);
  \coordinate (lo) at (-1.5,1.5); 
   \coordinate (lr) at (0,0); 
  \coordinate (ro) at (1.5,1.5); 
  \coordinate (c) at (0,1.1); 
  \coordinate (t1) at (-1.2,1.42); 
  \coordinate (t4) at (-1,1.36);
    \coordinate (t2) at (-.8,1.31); 
    \coordinate (t5) at (-.58,1.246);
        \coordinate (t3) at (-.37,1.186);
        \coordinate (t6) at (-.17,1.141);
           \fill[cyan!50] (l) -- (o) -- (0,0) -- (l);
   \draw (l) -- (r) node[midway, below]{$[x_1:x_2:0]$};
   \draw (l) -- (o) ;
   \draw (r) -- (o) ;
      \coordinate (u2) at (-2.5,0);
      \coordinate (u3) at (-2,0);
      \coordinate (u4) at (-1.5,0);
      \coordinate (u6) at (-1.1,0);
      \coordinate (u7) at (-0.7,0);
      \coordinate (u8) at (-0.3,0);
   \draw[blue] (o) -- (u2);
    \draw[blue] (o) -- (u3);
   \draw[blue] (o) -- (u4);
   \draw[blue] (o) -- (u6);
   \draw[blue] (o) -- (u7);
   \draw[blue] (o) -- (u8);
     \draw[line width=0.5mm, red] (lo) -- (c);
   \draw[line width=0.5mm, red] (ro) -- (c);
  \filldraw[fill=red] (lo) circle (0.07);
 \filldraw[fill=red] (ro) circle (0.07);
  \filldraw[fill=red] (c) circle (0.07);
 \filldraw[fill=black] (l) circle (0.07);
 \filldraw[fill=black] (r) circle (0.07);
 \filldraw[fill=black] (o) circle (0.07) node[above]{$[0:0:1]$};
   \filldraw[fill=blue!66] (0,0) circle (0.05);
  \end{tikzpicture} 
   \caption{A graph $G$ with a tadpole. Here the spine has two 1-cubes $C_T$ for $T\in \m T(G)=\set{e_1,e_2}$. The blue lines indicate the fibers over the cube $C_{e_1}$, connecting its points to the faces at infinity, the point $\set{x_1=x_2=0}$ and the line $\set{x_3=0}$.}\label{fig:fiberbundletadpole}
\end{figure}
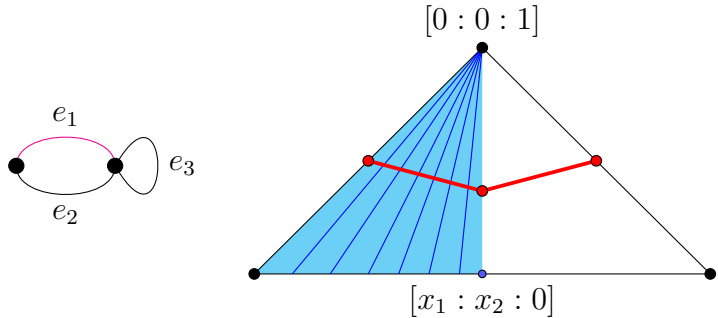

\subsection{Other fibrations}
The construction laid out here is of course by no means unique. Firstly, the spine $S_G$ can obviously be embedded in $\ol \Delta_G$ in many ways. However, our choice is in some sense canonical. Consider for instance the sunrise graph (\cref{fig:fiberbundlesunrise}). The space of edge lengths without normalization is a 3-dimensional cube $\mb R_+^3$. Working on the closure we can replace it by the unit cube $[0,1]^3$. The graph simplex $\ol \Delta_G$ is the link of the origin. If we embed it as the intersection of $[0,1]^3$ with the unit sphere (e.g.\ with respect to the 1-norm $||\cdot||_1$ to obtain an euclidean simplex), then the spine $C_G$ emerges naturally as the projection of the codimension one ``back faces" that connect to the point $(1,1,1)$, see \cref{fig:cube}. This construction generalizes in a straightforward way to every graph simplex and also to the full moduli space of graphs (or tropical curves); see \cite{berger-kaufmann:WBK}.

There are also many ways of expressing $\Delta_G$ as the total space of a fibration over its spine. Consider again the sunrise graph in \cref{fig:fiberbundlesunrise}. In this case it is easy to imagine different bundles that cover $\Delta_G$ up to a set of measure zero. It is however not clear how to make these ideas work in general, except using the geometric approach presented here, projecting the spine into the faces at infinity.

\section{A reminder on fiber integration}\label{sec:fiberintegrals}

We recall the notion of push-forward or integration along the fibers of a bundle, following \cite{ni}. See also \cite{bott-tu}.

Fiber integration is typically defined for bundles with compact fibers or as a map on compactly supported differential forms (or differential forms with compact support in the fiber direction). In contrast, the bundles that we consider have non-compact fibers. In addition, the Feynman forms $\omega_G$ that we want to integrate are not compactly supported, neither in the total space nor in the fiber. We are saved, though, by Fubini's theorem and the fact that the total integral $\int_{\sigma_G}\omega_G$ is absolutely convergent (when suitably regularized or renormalized).
Alternatively, we could work on the compactified cell $\ti \Delta_G$ in order to have the fibers of $\pi_T$ compact.

In any case we use the following proposition as a definition for fiber integration.

Let $\Omega_{\mathrm{cpt}}(X)$ denote the space of compactly supported forms on a smooth manifold $X$. We cite proposition 3.4.47 from \cite{ni}:

\begin{prop}
Let $p \colon E \to B$ be an orientable fiber bundle with fiber $F$, $k=\dim F$. Then there exists a linear operator 
\[p_*\colon \Omega^\bullet_{\mathrm{cpt}}(E)  \to \Omega^{\bullet-k}_{\mathrm{cpt}}(B).\]

It is uniquely defined by its action on forms supported on local trivializations where $p \colon \mb R^k \times \mb R^n \to \mb R^n$.
If $\omega = f dx^I \wedge dy^J$ with $f \in  C_0^\infty(\mb R^k \times \mb R^n)$, then
\[
p_*\omega = \begin{cases} 0 & \text{ if }  |I| \neq k, \\
\left( \int_{\mb R^k} fdx^I \right) dy^J & \text{ if } |I|=k.
\end{cases}
\]
\end{prop} 
The operator $p_*$ is called \emph{push-forward along $p$} or \emph{integration along the fibers of $p$}.

\section{Feynman integrals}\label{sec:FI}
We apply the results of the previous two sections to parametric Feynman integrals. We start with a quick review on parametric Feynman integrals. A thorough discussion can be found in \cite{erik:phd}.

\subsection{Feynman integrals in Schwinger parameters} \label{ssec:parametricFI}

Let $G$ be a (massive scalar) Feynman diagram, that is, a finite connected graph with $n$ labeled edges (masses $m_e$ and powers $a_e$ of propagators) and $k$ labeled legs (momenta $p_i$ of in/out-going particles). The Feynman integral of $G$ is (up to a prefactor of constants and $\Gamma$-functions) 
\begin{equation}\label{eq:FIparametricdetail}
    I_G(p,m,a,D)=\int_{\Delta_G} \omega_G
\end{equation}
where $p=(p_1,\ldots,p_k)$ with $\sum_{i=1}^k p_i=0$, $m=(m_1,\ldots,m_n)$, $a=(a_1,\ldots,a_n)$ and 
\begin{equation*}
    \omega_G = \psi_G^{-\frac{D}{2}} \left(\frac{\varphi_G}{\psi_G}\right)^{\mathrm{sdd}(G)}\prod_{i=1}^n x_i^{a_i-1} \cdot \Omega_G, \quad \Omega_G= \sum_{i=1}^n (-1)^ix_i dx_1 \wedge \ldots \wedge \widehat{dx_i} \wedge \ldots \wedge dx_n.
\end{equation*}
Here $\psi_G$ and $\varphi_G$ are the two graph polynomials
\begin{align*}
    \psi_G & = \sum_{T\in \m T(G)} \prod_{e\notin E_T} x_e, \\
    \varphi_G & = \psi_G \sum_{i=1}^n m_e^2x_e + \sum_{\substack{F=T_1\sqcup T_2 \\ T_i \in \m T(G)}} p(T_1)^2 \prod_{e\notin E_F} x_e,
\end{align*}
$p(T_1)=-p(T_2)$ denoting the sum of all momenta flowing into the $T_1$-component of the 2-forest $F$. The complex number 
\begin{equation*}
    \mathrm{sdd}(G) = \frac{D}{2}h_1(G) - \sum_{i=1}^n a_i 
\end{equation*}
is the superficial degree of divergence of $G$. As the name suggests, it determines the region of parameters $(p,m,a,D)$ where $I_G$ converges absolutely. 

\begin{thm}[Weinberg \cite{weinberg}]
If all $m_e$ are positive and $\mathrm{sdd}(\gamma)<0$ holds for all 1-PI/core subgraphs $\gamma \subset G$, then $I_G$ converges absolutely.
\end{thm}

We henceforth assume that the parameters $(a,D)$ lie in the region determined by this criterion. Otherwise we would have to renormalize the integrand \cite{bk-as} in order to make sense of the integral. Since we are mainly interested in structural properties we omit a discussion of this procedure\footnote{See \cite{mb2} for a discussion of Feynman amplitudes and renormalization on the moduli space of graphs and its compactification $\ti MG_g$.}. The main message is that the integral $I_G$ converges absolutely, so that we can apply the aforementioned formalism of fiber integration to this case.

\subsection{Loop-tree duality} \label{ssec:parametricLTD}

To arrive at a parametric variant of loop-tree duality we decompose the integration domain of $I_G$ as $\Delta_G = 
\bigcup_{T \in \m T (G)} D_T $, using \cref{prop:covering}. 

We can thus rewrite $I_G$ as 
\[
I_G= \sum_{T \in \m T(G)} I_{G,T} \  \text{ where } \ I_{G,T} = \int_{D_T} \omega_G.
\]
The individual integrals $I_{G,T}$ can be further simplified, 
\[
I_{G,T}= \int_{D_T} \omega_G \overset{\text{a.e.}}{=} \int_{C_T} ({\pi_T})_* \omega_{G}.
\]
The second identity is valid only outside a set of measure zero $X_T \subset D_T$ where the fiber bundle $\pi_T$ is not smooth (see \cref{prop:fiberbundle}). To be precise we should define $\hat D_T= D_T - X_T $ and
\be \label{eq:hatpi}
 \hat \pi_T = {\pi_T}\big|_{\hat D_T} \colon \hat D_T \longrightarrow  C_T
\ee
so that we can write 
\[
\int_{D_T} \omega_G = \int_{\hat D_T} \omega_G  = \int_{C_T} (\hat \pi_T)_* \omega_G.
\]

It is important to note that there is no map $(\hat \pi_T)_*\colon \Omega^\bullet(\hat D_T) \to \Omega^{\bullet - h_G + 1 }(C_T)$; the pushforward is only defined on those forms for which the fiber integral is finite\footnote{It exists though as a map $(\ti \pi_T)_*\colon \Omega^\bullet(\ti D_T - X_T) \to \Omega^{\bullet - h_G + 1 }(C_T)$ for $\ti D_T$ the closure of $D_T$ in $\ti \Delta_G$ and $X_T$ as above.}, in particular those forms $\omega \in \Omega^\bullet(\hat D_T)$ for which $\int_{\hat D_T}\omega$ converges absolutely. 

Putting everything together, we have proved

\begin{thm}[Parametric loop-tree duality]
Every (convergent/renormalized) parametric Feynman integral
\[
I_G=\int_{\sigma_G} \omega_G
\]
as in \cref{eq:FIparametricdetail} can be written as a sum of integrals over unit cubes $C_T\cong (0,1)^{e_T}$, indexed by the set $\m T(G)$ of spanning trees of $G$,
\be \label{eq:finalparametricLTD}
I_G = \sum_{T \in \m T(G)} \int_{C_T} \omega_{G,T},
\ee
where $\omega_{G,T} = (\hat \pi_T)_*\omega_G \in \Omega^{e_T}(C_T)$ and $\hat \pi_T$ the smooth fiber bundle defined above in \cref{eq:hatpi}.
\end{thm}

\begin{rem}[The one loop case]
If $h_G=1$, then the cubes $C_T$ already partition the simplex $\Delta_G$,
\[
\Delta_G= \bigcup_{e \in E_G} \set{ x_e > x_{e'} \text{ for all }  e'\neq e }.
\]
Thus in this case \cref{eq:finalparametricLTD} takes the form
\[
I_G= \sum_{T \in \m T(G)} \int_{C_T} (\omega_G)_{|C_T}.
\]
\end{rem}

\subsection{Examples}\label{ssec:examples}

Recall the discussion in \cref{eg:sunrise}. For $T=e_1$ we have in the affine chart with $x_3=1$
\[
D_T=\set{ (z,w) \in (0,1)\times \mb R_+ \mid 0<z<\min(w,1)}, \quad C_T = \set{ (z,1) \mid 0<z<1 }
\]
and
\[
\pi_T \colon D_T \longrightarrow C_T, \  (z,w) \longmapsto  \left(\frac{z}{\min(w,1)},1\right)= \begin{cases} (z,1) & \text{ if } w>1 ,\\
(\frac{z}{w},1) & \text{ if } w<1, \end{cases}
\]
so that
\[
\inv{\pi_T}(u,1)= (\set{u} \times (1,\infty)) \cup \set{ (uv,v) \mid  v \in (0,1)} \subset D_T.
\]
This is smooth on the complement of $C_T$ in $D_T$, so $\hat D_T=D_T - C_T$. In these coordinates the integral over the sector $D_T$ is given by 
\[
I_{G,T}=\int_{D_T} \omega_G=  \int_{ (0,1) \times  (1,\infty) \cup \set{ (uv,v) \mid  u,v \in (0,1)} }  f_G(z,w)\, dzdw
\]
with $f_G$ determined by the Feynman rules in \cref{ssec:parametricFI}.
Then the fiber integral is
\[
    \omega_{G,T}=  \left( \int_{(1,\infty)} f_G(z,w)\, dw +    \int_{(0,1)} f_G(zw,w)w\, dw  \right) dz
\]
and \cref{eq:finalparametricLTD} takes thus the form
\[
I_G = 3 \cdot \int_{(0,1)} \left(  \int_{(1,\infty)} f_G(z,w)\, dw +    \int_{(0,1)} f_G(zw,w)w \, dw \right) dz.
\]
See also \cite[Appendix E]{dirk:bananas} which discusses the renormalized parametric sunrise integral along similar lines.

To be more concrete let us consider a simple toy model. If 
\[
\omega_G=\frac{x_3\Omega_G}{(x_1 + x_3)^2(x_2 + x_3)^2},
\]
then for $T_i=\set{e_i}$ we get
\begin{align*}
    \omega_{G,T_1}=\omega_{G,T_2} &= \left( \int_{(1,\infty)} \frac{dw}{(w + 1)^2(z + 1)^2}  +    \int_{(0,1)} \frac{w\, dw}{(zw + 1)^2(w + 1)^2}  \right) dz ,\\
    \omega_{G,T_3} & = \left( \int_{(1,\infty)} \frac{z \, dw}{(w + z)^2(z + 1)^2}  +    \int_{(0,1)} \frac{z \, dw}{(zw + 1)^2(z + 1)^2}  \right) dz ,
\end{align*}
and thus
\[
I_G= \int_{\Delta_G} \omega_G = \int_{(0,\infty)^2} \frac{dx_1 dx_2}{(x_1 + 1)^2(x_2 + 1)^2}  = 2 \int_0^1\omega_{G,T_1}  + \int_0^1 \omega_{G,T_3}
\]
which may be confirmed by direct integration.

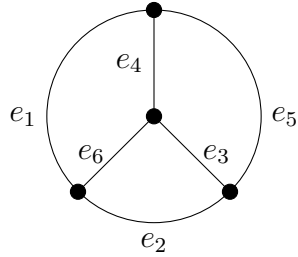
\begin{figure}[ht]
\begin{tikzpicture}[scale=1]
\draw (0,0) circle (1.41cm);
\coordinate (c) at (0,0);
\coordinate (l) at (-1,-1);
\coordinate (r) at (1,-1);
\coordinate (o) at (0,1.41);
 \draw (l) to node[left] {$e_6$} (c);
 \draw (r) to node[right] {$e_3$} (c);
 \draw (o) to node[left] {$e_4$} (c);
  \filldraw[fill=black] (l) circle (0.1);
  \filldraw[fill=black] (r) circle (0.1);
  \filldraw[fill=black] (o) circle (0.1);
  \filldraw[fill=black] (c) circle (0.1);
  \node (ll) at (-1.4,0) [left] {$e_1$};
  \node (rr) at (1.4,0) [right] {$e_5$};
  \node (oo) at (0,-1.41) [below] {$e_2$};
 \end{tikzpicture}
  \caption{The wheel with three spokes.}\label{fig:wheel}
\end{figure}

As another example we consider the wheel with three spokes (\cref{fig:wheel}). Its \textit{period} is the integral 
\[
\m P(G) = \int_{\Delta_G} \frac{\Omega_G}{\psi_G^2}
\]
with
\begin{align*}
\psi_G = & x_1x_2x_5 + x_1x_2x_3 + x_1x_2x_4 + x_1x_3x_5 + x_1x_5x_6 + x_1x_3x_4 + x_1x_3x_6 + x_1x_4x_6
\\
& + x_2x_4x_5 + x_2x_5x_6 + x_2x_3x_4 + x_2x_3x_6 + x_2 x_4x_6 + x_3x_4x_5 + x_3x_5x_6 + x_4x_5x_6.
\end{align*}
The spanning trees of $G$ come in two different shapes, either paths on three edges, for instance $\set{e_1,e_2,e_3}$, or ``claws"\footnote{In graph theory, a \textit{claw} is a star graph on three edges; a star graph on $k$ edges is a tree with one internal vertex and $k$ leaves.}, for instance $\set{e_1,e_2,e_6}$. In both cases the set $\mg$ consists of the core subgraphs obtained from adding to $T$ two out of the three edges from its complement.

For example, if $T=\set{e_1,e_2,e_3}$, then 
\[
\mg= \Big\{ \gamma \subset G \colon E_\gamma=\set{e_1,\ldots,\widehat{e_i},\ldots,e_6 } \mid i=4,5,6 \Big\}
\]
and therefore 
\begin{align*}
    \inv{\pi_T}(x)  =& \set{ [x_1:x_2:x_3:\lambda x_4:\mu x_5:x_6] }  \cup \set{ [x_1:x_2:x_3:\lambda x_4: x_5:\mu x_6] }  \\
 &  \cup \set{ [x_1:x_2:x_3: x_4: \lambda x_5:\mu x_6] }, \quad \lambda,\mu > 1 .
\end{align*}
In the affine chart $x_6=1$ where
$
C_T= \set{ (x_1,x_2,x_3,1,1) \mid x_i \in (0,1) }
$
we find 
\begin{align*}
 D_T= & \set{ (x_1,x_2,x_3,u,v) \mid u,v > 1 } \cup \set{ (vx_1,vx_2,vx_3,vu,v) \mid u>1,v\in (0,1) }   \\
 &  \cup \set{(vx_1,vx_2,vx_3,v,vu)\mid u>1,v\in (0,1)}.
\end{align*}
With $f_i=\psi_G^{-2}\big|_{x_i=1}$ we have then
\begin{align*}
\omega_{G,T}  = & \bigg( \int_1^\infty dv \int_1^\infty du \, f_6(x_1,x_2,x_3,u,v)    + \int_0^1 dv \int_1^\infty  du \, v^4 f_6(vx_1,vx_2,vx_3,vu,v) \\
&+  \int_0^1 dv \int_1^\infty du\,  v^4 f_6(vx_1,vx_2,vx_3,v,vu) \bigg) dx_1dx_2dx_3.
\end{align*}

We can explicitly compute the forms $\omega_{G,T}$ by means of the \textit{contraction-deletion identity} 
\[
\psi_G=\psi_{G\backslash e_i}x_i + \psi_{G/e_i}.
\]
Write $\psi^I_J$ for $\psi_{G'}$ where $G'$ is the graph $G$ with edges $e_i$, $i\in I$, removed and $e_j$, $j\in J$, contracted. Then $\psi^i x_i + \psi_i$ implies  $f_j=(\psi^j+\psi_j)^{-2}$. Furthermore, if $\lambda_i$ denotes the map 
\[
x_k \longmapsto \begin{cases} x_i x_k & \text{ if }  k\neq i,\\ x_i & \text{ else,}
\end{cases}
\]
then 
\[
f_j\circ \lambda_i= \frac{1}{ x_i^{4}\Big(  (\psi^i_j + \psi_{ij})x_i + (\psi^{ij} + \psi^j_i)  \Big)^2 }.
\]
 It follows that
\begin{align*}
\omega_{G,T}  = & \bigg( \int_1^\infty dx_5 \int_1^\infty dx_4  \frac{1}{ \big( (\psi^5_6 + \psi^{56})x_5 + \psi_{56} + \psi^6_5 \big)^2 } \\
 &+\int_0^1 dx_5 \int_1^\infty  dx_4  \frac{1}{ \big( (\psi^5_6 + \psi_{56})x_5 + \psi^{56} + \psi^6_5  \big)^2 } \\
 & +\int_0^1 dx_4 \int_1^\infty dx_5   \frac{1}{\big(  (\psi^4_6 + \psi_{46})x_4 + \psi^{46} + \psi^6_4  \big)^2} \bigg) dx_1dx_2dx_3 \\
 = & \bigg(  \int_1^\infty dx_4 \frac{1}{ \psi^5_6 + \psi_{56}  + \psi^{56} + \psi^6_5 }  \cdot  \left( \frac{1}{\psi^5_6 + \psi^{56}} + \frac{1}{\psi^6_5 + \psi^{56}} \right)   \\
&+   \int_1^\infty dx_5   \frac{1}{( \psi^{46} + \psi^6_4 ) (\psi^4_6 + \psi_{46} + \psi^{46} + \psi^6_4) } \bigg) dx_1dx_2dx_3
\end{align*}
which evaluates to a rational linear combination of logarithms with rational arguments in $(x_1,x_2,x_3)\in (0,1)^3$.

\section{Outlook}\label{sec:outlook}

\subsection{Iterated fiber integrals}
The fibers $\inv{\pi_T}(x)$ are unions of open cones. The closure of this union in $\Delta_G$ is homeomorphic to a simplex of the same dimension. 
 Each such simplex can again be expressed as the union of total spaces of fiber bundles over cubes.
We may therefore iterate the process of fiber integration. 

Let us look at the $4$-edge banana graph $B_4$ in \cref{fig:4banana}. The closure in $\Delta_{B_4}$ of the fiber over $x \in C_{e_4}$ is (PL) homeomorphic to a 2-simplex which we identify with the space of metrics on the graph $G-T=G-e_4$, a 3-edge banana $B_3$. Moreover, the three lines from $x$ to the corners at infinity on which $\pi_{e_4}$ is not smooth correspond precisely to the three lines where $\pi_{B_3}$ fails to be a fiber bundle (see \cref{fig:sunrisecompact}). This means we can express the integral along the fibers of $\pi_{B_4}$ itself as a (sum of) fiber integrals of the form $I_{B_3}$, up to a set of measure zero.

For a general banana graph we can iterate this process until we are left with one-dimensional fibers, that is, until $G-T_1-T_2 - \ldots - T_k$ is the 3-edge banana graph. Here, no further simplification is possible.

Note that this is very reminiscent of the recursive structure used in \cite{dirk:bananas} to discuss the analytic structure of Feynman integrals of banana graphs.

It would be interesting to study if and how this generalizes to graphs with banana subgraphs (that is, graphs with multi-edges).

\subsection{Periods}
The loop-tree duality formula \cref{eq:finalparametricLTD} arises solely from a decomposition of the integration domain subject to the combinatorics of the graph $G$.

If the integrand is sufficiently nice, for instance in the case of periods, or \textit{generalized periods} in the sense of \cite{fb:periods}, it may be possible to further simplify the forms $\omega_{G,T}$ by exploiting the structure of the graph polynomial $\psi_G$ (cf.\ \cite[\S 2,\S3]{fb:periods}). The example of the wheel with three spokes shows that we can always do one fiber integration using the contraction-deletion identity for $\psi_G$. 

A natural question to ask then is for which forms $\omega_G$ the pushforwards $\omega_{G,T}$ can be computed explicitly and how does it depend on the combinatorics of $G$ and $T$.

\subsection{Another parametric representation}

Consider the equivalent representation 
\[
I_G(p,m,a,D) = \int_{(0,\infty)^{n}}  \frac{e^{ -\frac{\varphi_G(p,m)}{\psi_G}}}{\psi_G^{\frac{D}{2}}} \prod_{i=1}^n x_i^{a_i-1} dx_i
\]
where we integrate over the full cell $\sigma(G)$ in ${\mathcal{MG}}_g$, that is, without normalizing the volume of $G$. The decomposition of $\Delta_G$ that we constructed above can be ``lifted" to a decomposition of $\sigma(G)$: The spine becomes a fan based at the origin $0\in \mb R^n$, the maps $\pi_T$ assemble to one-dimensional families of fiber bundles. 

\begin{figure}[ht]
\begin{tikzpicture}[scale=1.5]
 \coordinate (o) at (0,0,0);
 \coordinate (e1) at (2,0,0.4);
 \coordinate (e2) at (0,2,0);
 \coordinate (e3) at (-0.4,0,2);
 \coordinate (e12) at (2,2,0.4);
 \coordinate (e13) at (1.6,0,2.4);
 \coordinate (e23) at (-0.4,2,2);
 \coordinate (e123) at (1.6,2,2.4);
\filldraw[gray, opacity=0.3] (0,1,0) -- (1,0,0.2) -- (-0.2,0,1) -- (0,1,0);
\filldraw[red, opacity=0.3] (o) -- (e23) -- (e123) -- (o);
\filldraw[red, opacity=0.4] (o) -- (e13) -- (e123) -- (o);
\filldraw[red, opacity=0.3] (o) -- (e12) -- (e123) -- (o);
\draw[dashed] (o) -- (e1);
\draw[dashed] (o) -- (e2);
\draw[dashed] (o) -- (e3);
\draw (e13) -- (e3);
\draw (e23) -- (e3);
\draw (e2) -- (e23);
\draw (e2) -- (e12);
\draw (e1) -- (e12);
\draw (e1) -- (e13);
\draw[black!50] (e13) -- (e123);
\draw[black!50] (e123) -- (e23);
\draw[black!50] (e123) -- (e12);
\draw[gray] (0,1,0) -- (1,0,0.2);
\draw[gray] (0,1,0) -- (-0.2,0,1);
\draw[gray] (-0.2,0,1) -- (1,0,0.2);
\coordinate (c1) at (-0.1,0.5,0.5);
\coordinate (c) at (0.4,0.5,0.6);
\coordinate (c2) at (0.4,0,0.6);
\coordinate (c3) at (0.5,0.5,0.1);
\draw[red, dashed] (c1) -- (c);
\draw[red, dashed] (c2) -- (c);
\draw[red, dashed] (c3) -- (c);
 \end{tikzpicture}
  \caption{The cell $\sigma(G)$ of the sunrise graph. The gray simplex is $\Delta_G$, its spine indicated by the red dashed lines . The three red cones form the associated fan in $\mb R_+^3$. To find the fiber over a point $x$ in one of these cones, translate $\Delta_G$, so that $x \in \tau.\Delta_G=\set{ \sum_{i=1}^3y_i=\tau}$, then take the corresponding fiber in $\tau.\Delta_G$ (as in \cref{fig:fiberbundlesunrise}).}\label{fig:cube}
\end{figure}
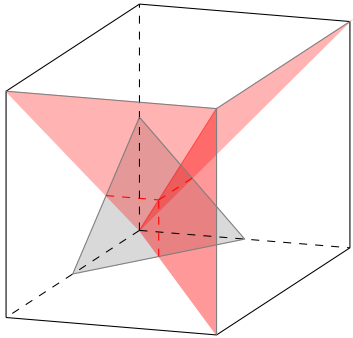

In this representation the problem of ultraviolet and infrared renormalization splits in an interesting way. While the divergences at $x_e \to 0$ are still located in the fibers of $\pi_T$, the divergences at $x_e \to \infty$ (see e.g.\ \cite{AH_Miz:IRpoly}) lie now on the respective components of the spine. This means that ultraviolet renormalization is needed to define the (lifted) push-forwards $(\pi_T)_*$ while the infrared divergences appear then at the level of the forms $\omega_{G,T}$.

\subsection{Differential forms on graph complexes}

The parametric loop-tree duality formalism applies of course also to other differential forms on $\Delta_G$, not just those given by Feynman rules. Moreover, one may also try to apply it to forms on the (full) moduli space of graphs or tropical curves, for instance the \textit{canonical forms} constructed in \cite{brown21,brown22}. 

Recall the definition of $MG_g$ in \cref{ssec:modspace}. A \textit{smooth differential form of degree $k$ on $MG_g$} is a family
\[
 \omega = \big\{ \omega_G \in \Omega^k( \Delta_G) \mid G \in \mathrm{ob(\Gamma_g)} \big\}
\]
such that
\begin{enumerate}
    \item if $\varphi\colon G \overset{\sim}{\to} G'$ is an isomorphism and $\Delta_\varphi$ the induced map $ \Delta_G \overset{\sim}{\to} \Delta_{G'}$, then $\Delta_\varphi^*\omega_{G'}=\omega_G$,
    \item if $e\in E_G$ is not a tadpole and $\iota_e$ is the inclusion $\Delta_{G/e} \hookrightarrow \ol \Delta_G$, then $\omega_G$ extends to the face $\iota_e(\Delta_{G/e})$ and satisfies $\iota_e^*  \omega_G = \omega_{G/e} $.
\end{enumerate}

Examples can be found in \cite{brown21,brown22}. Note that Feynman rules do not define differential forms on $MG_g$; they define a form $\omega_G$ on $\Delta_G$ for each $G$, but these do not assemble to a differential form in the above sense. Instead one may interpret the family $\set{\omega_G}_G$ as a distribution on the cell complex $MG_g$ \cite{mb2}.

Given a differential form $\omega$ on $MG_g$ such that $\int_{\Delta_G}\omega_G$ converges absolutely for every graph $G$, we can associate to it a differential form $\eta$ on the maximal cubes of the spine $S_g$. On a cube $C_T^G\subset \Delta_G$ (we add a superscript label to cubes and maps in order to keep track of the cell we are working in) where $T\in \m T(G)$ use the push-forward along the corresponding restriction $\pi_T^G$ of $\pi \colon MG_g \to S_g$ to set 
\[
\eta_T^G = (\pi_T^G)_*\omega_G.
\]

Now the question is if and how $\eta$ can be extended to the full spine $S_g$. 
A cube $\ol {C}_T^G$ has two kinds of faces; they correspond either to collapsing an edge $e$ in $T$, or removing it from $T$. 

In the first case, if we send $x_e$ to zero, we get a cube $C_{T/e}^{G/e}$ in $\Delta_{G/e}$. By functoriality of both $\omega$ and $\pi$ (item (2) above and \cref{rem:homotopyequivalence}) we have
\[
\iota_e^*  \eta_T^G = \eta_{T/e}^{G/e}.
\]

In the second case the push-forward along $\pi$ is not defined, since over $C_{T-e}^G$ the map $\pi$ is not a fiber bundle. If we simply set 
\[
\eta_{T-e}^G = \tau_{e}^* \eta_T^G
\]
where $\tau_{e}$ is the inclusion $C_{T-e}^G \hookrightarrow C_T^G$ defined by sending $x_e$ to 1, then it is not clear whether this is well-defined. In fact, examples show that it is not. However, if $\eta_T^G$ and $\eta_{T'}^G$ are closed, and for $T-e=T'-e'$ the differences $\tau_{e}^* \eta_T^G - \tau_{e'}^* \eta_{T'}^G=d\alpha$ are exact, then (under some mild technical assumptions) the triple $(\alpha, \eta_T^G , \eta_{T'}^G)$ defines a cohomology class on $C_T \cup C_{T'} \cup C_{T-e}$.\footnote{Thanks to Erik Panzer for pointing this out to me.} 

These ideas and their application to the study of graph complexes will be pursued in future work.


\printbibliography

\end{document}